\documentclass[11pt]{llncs}
\usepackage{amssymb,amsmath,enumerate,graphicx,hyperref,amsfonts,latexsym,yhmath,url}

\usepackage{tikz}
\usetikzlibrary{snakes,shapes}
\tikzstyle{hyb}=[rectangle,fill=green!50,draw,minimum size=3mm]
\tikzstyle{tre}=[circle,fill=green!50,draw,minimum size=3.5mm]

\renewcommand{\leq}{\leqslant}
\renewcommand{\geq}{\geqslant}
\newcommand{\NN}{\mathbb{N}}

\newcommand{\TT}{\mathcal{T}}

\newcommand{\RR}{\mathbb{R}}
\newcommand{\oF}{\overline\Phi}
\newcommand{\Cov}{\textit{cov}}
\newcommand{\mmod}{\hspace*{-1.5ex}\mod}
\newcommand{\ee}{\varepsilon}
\newcommand{\nosum}{{\cal S}}

\begin{document}

\title{Exact formulas for the variance of several balance indices under the Yule model}

\author{Gabriel Cardona\and Arnau Mir\and Francesc Rossell\'o}

\institute{Dept. of Mathematics and Computer Science, University of the Balearic Islands, E-07122 Palma. \email{\{gabriel.cardona,arnau.mir,cesc.rossello\}@uib.es}%
}
\maketitle

\begin{abstract}
One of the main applications of balance indices is in tests of null models of evolutionary processes. 
The knowledge of an exact formula for a statistic of a balance index, holding  for any number $n$ of leaves,  is necessary in order to use this statistic in tests of this kind involving trees of any size. In this paper we obtain exact formulas for the variance under the Yule model of the Sackin, the Colless and the total cophenetic indices of binary rooted phylogenetic trees with $n$ leaves. 
\end{abstract}

\section{Introduction}
\label{sec:intro}

One of the most thoroughly studied properties of the topology of phylogenetic trees is their symmetry, that is, the degree to which both children of each internal node tend to have the same number of descendant taxa.  The symmetry of a tree is usually quantified by means of \emph{balance  indices}. Many such indices have been proposed so far in the literature  \cite[Chap. 33]{fel:04}. The most popular are \emph{Colless' index} $C$  \cite{Colless:72}, which is defined as the sum, over all internal nodes $v$,  of the absolute value of the difference between the number of descendant leaves of $v$'s children, and \emph{Sackin's index} $S$  \cite{Sackin:72}, which is defined as the sum of the depths of all leaves in the tree.
 We have recently proposed an extension of Sackin's index, the \emph{total cophenetic index} $\Phi$ \cite{MRR}: the sum, over all pairs of different leaves of the tree, of the depth of their least common ancestor. The main advantages of $\Phi$ over $S$ are that it has a larger range of values and a smaller probability of ties. Moreover, $\Phi$ retains other good properties of $S$: it makes sense for not necessarily fully resolved phylogenetic trees (unlike  Colless' index), it can be computed in linear time, and the statistical properties of its distribution of values can be studied under different stochastic models of evolution, like for instance the Yule \cite{Harding71,Yule} and the uniform \cite{CS,Rosen78,cherries} models. This last property is relevant because one of the main applications of balance indices is their use as tools to test stochastic models  of evolution \cite{Mooers97,Shao:90}. 

Exact formulas for the expected values under the Yule model of $C$, $S$, and $\Phi$ on the space $\TT_n$ of fully resolved rooted phylogenetic trees with $n$ leaves have already been  published. More specifically, if we denote by $H_n$ the $n$-th \emph{harmonic number},  
$$
H_n=\sum_{i=1}^n \frac{1}{i},
$$
these expected values are, respectively,
$$
\begin{array}{ll}
E_Y(C_n)=(n\hspace*{-1ex} \mod 2)+n(H_{\lfloor \frac{n}{2}\rfloor}-1)
 & \mbox{  \cite{Heard92}}\\
E_Y(S_n)=2n(H_n-1) & \mbox{   \cite{KiSl:93}}\\
 E_Y(\Phi_n)=n(n-1)-2n(H_n-1)& \mbox{  \cite{MRR}}
\end{array}
$$

As we have already pointed out in \cite{MRR}, the last two formulas imply that the expected value under the Yule model of the sum $\oF=S+\Phi$ on $\TT_n$ is 
$$
E_Y(\oF_n)=n(n-1),
$$
a quite simpler expression than those for $E_Y(S_n)$ or $E_Y(\Phi_n)$.
This index $\oF$ has the same good properties of $\Phi$, but the formulas for its statistics under the Yule model tend to be simpler than the corresponding formulas for other indices. We shall find here another example of this fact: the variance.

The goal of this paper is to provide exact formulas for the variance of $S$, $C$, $\Phi$ and $\oF$  on $\TT_n$ under the Yule model. As a byproduct of our computations, we shall also obtain the covariances of $S$ with $\Phi$ and $\oF$. The variances of $S$ and $C$ on $\TT_n$ under this model were known so far only  for their limit distribution when $n\to \infty$ \cite{BFJ:06}:
$$
\begin{array}{l}
\sigma_Y^2(C_n) 
\sim \Big(3-\dfrac{\pi^2}{6}-\log(2)\Big)n^2\\[2ex]
\sigma_Y^2(S_n) 
\sim \Big(7-\dfrac{2\pi^2}{3}\Big)n^2
\end{array}
$$
Also, Rogers \cite{Rogers:94,Rogers:96} found  recursive formulas for the moment-generating functions of $C$ and $S$ under this model, which allow one to compute recursively as many values of $\sigma_Y^2(C_n)$ and $\sigma_Y^2(S_n)$ as desired, but he did not produce explicit exact formulas for them.

In this paper we obtain the following exact formulas for these variances:
$$
\begin{array}{l}
\sigma_Y^2(C_n^2)  \displaystyle  =\frac{5n^2+7n}{2}+(6n+1)\Big\lfloor\frac{n}{2}\Big\rfloor
-4\Big\lfloor\frac{n}{2}\Big\rfloor^2+8\Big\lfloor\frac{n+2}{4}\Big\rfloor^2\\
\qquad\displaystyle -8(n+1)\Big\lfloor\frac{n+2}{4}\Big\rfloor
-6nH_n+\Big(2\Big \lfloor\frac{n}{2}\Big\rfloor - n(n-3)\Big)H_{\lfloor\frac{n}{2}\rfloor}\\
\qquad\displaystyle -n^2H_{\lfloor\frac{n}{2}\rfloor}^{(2)}+\Big(n^2+3n-2\Big\lfloor\frac{n}{2}\Big\rfloor\Big)H_{\lfloor\frac{n+2}{4}\rfloor}-2nH_{\lfloor\frac{n}{4}\rfloor}
\\[2ex]
\displaystyle \sigma^2_Y(S_n)=7n^2-4n^2H_n^{(2)}-2nH_n-n\\[1ex]
\displaystyle\sigma^2_Y(\Phi_n)= \frac{1}{12}(n^4-10n^3+131n^2-2n)-6nH_n-4nH_n^2-4n(n-1)H_n^{(2)}\\[2ex]
\displaystyle\sigma^2_Y(\oF_n)=2\binom{n}{4}
\end{array}
$$
where $H_n^{(2)}=\sum\limits_{i=1}^n 1/i^2$.
We also obtain the following exact formulas for the covariances, under the Yule model, of $S$ with $\Phi$ and $\oF$ on $\TT_n$:
$$
\begin{array}{l}
\displaystyle\Cov_Y(S_n, \Phi_n)=  4n(nH_n^{(2)}+H_n)+\frac{1}{6} n (n^2-51 n+2)\\[2ex]
\displaystyle \Cov_Y(S_n, \oF_n)=  2nH_n+\frac{1}{6} n (n^2-9 n-4)
\end{array}
$$
All these formulas are valid for any number $n$ of leaves, and therefore they can be used in a meaningful way in tests involving trees of any size. The proofs  consist mainly of elementary, although long and technically involved, algebraic computations.

The rest of this paper is organized as follows. In a first section on Preliminaries we gather some notations and conventions on phylogenetic trees and some lemmas on probabilities of trees under the Yule model and on harmonic numbers. In the next section, we establish a recursive formula for the expected value under the Yule model of the square of a balance index satisfying a certain kind or recursion 
(a \emph{recursive shape index}   \cite{Matsen}) that lies at the basis of all our computations. Then, we devote a series of sections to compute the variances of $S$, $\Phi$, $\oF$, $C$ and the covariance of $S$ with  $\Phi$ and $\oF$, respectively. These sections consist of  long and tedious algebraic computations, without any interest beyond the fact that they prove the formulas announced above. We close the paper with a section on Conclusions and Discussion.

\section{Preliminaries}
\label{sec:prel}

\subsection{Phylogenetic trees}
\label{subsec:phtrees}
In this paper, by a  \emph{phylogenetic tree} on a set $S$ of taxa we mean a binary rooted tree  with its leaves bijectively labeled in  the set $S$. We shall always understand a phylogenetic tree as a directed graph, with its arcs pointing away from the root. To simplify the language, we shall always identify a leaf of a phylogenetic tree with its label.  We shall also use the term \emph{phylogenetic tree with $n$ leaves} to refer to a phylogenetic tree on the set $\{1,\ldots,n\}$.  We shall denote by  $\TT(S)$ the set of isomorphism classes of phylogenetic trees on a set $S$ of taxa, and by $\TT_n$ the set  $\TT(\{1,\ldots,n\})$ of isomorphism classes of phylogenetic trees with $n$ leaves.  We shall denote by $V_{int}(T)$ the set of internal nodes of a phylogenetic tree $T$.

Whenever there exists a path from $u$ to $v$ in a phylogenetic tree $T$, we shall say that $v$ is a  \emph{descendant} of $u$ and that $u$ is an \emph{ancestor} of $v$.   The \emph{lowest common ancestor} $\textrm{LCA}_T(u,v)$ of a pair of nodes $u,v$ in a phylogenetic tree $T$  is the unique common ancestor of them that is a descendant of every other common ancestor of them.  

The \emph{depth} $\delta_T(v)$ of a node $v$  in $T$ is the length (in number of arcs) of the unique path from the root $r$ of $T$ to $v$.
The \emph{cophenetic value} $\varphi_T(i,j)$ of a pair of leaves $i,j$ is the depth of their lowest common ancestor  \cite{Sokal:62}:
$$
\varphi_T(i,j)=\delta_T(\textrm{LCA}_T(i,j)).
$$ 
To simplify the notations at some points, we shall also write $\varphi_T(i,i)$ to denote the depth $\delta_T(i)$ of a leaf $i$. 

Given two phylogenetic trees $T,T'$ on disjoint sets of taxa $S,S'$, respectively, their \emph{tree-sum} is the tree $T\,\widehat{\ }\, T'$on $S\cup S'$ obtained by connecting the roots of $T$ and $T'$ to a (new) common root.
Every tree with $n$ leaves is obtained as $T_k\widehat{\ }\, {}T'_{n-k}$, for some $1\leq k\leq n-1$, some subset  $S_k\subseteq \{1,\ldots,n\}$ with $k$ elements, some tree $T_k$ on $S_k$ and some tree $T'_{n-k}$ on $S_k^c=\{1,\ldots,n\}\setminus S_k$; actually, every tree $T$ with $n$ leaves is obtained in this way \emph{twice}.

The \emph{Yule}, or  \emph{Equal-Rate Markov}, model of evolution \cite{Harding71,Yule} is a stochastic model of phylogenetic trees' growth. It starts with a node, and at every step a leaf is chosen randomly and uniformly and it is splitted into two leaves. Finally, the labels are assigned randomly and uniformly to the leaves once the desired number of leaves is reached. This corresponds to a model  of evolution where, at each step, each currently extant species can give rise,  with the same probability, to two new species. Under this model of evolution, different trees with the same number of leaves may have different probabilities. More specifically, if $T$ is a phylogenetic tree with $n$ leaves, and for every internal node $v$ we denote by $\kappa_T(v)$ the number of its descendant leaves, then the probability of $T$ under the Yule model is   \cite{Brown,SM01}
 \begin{equation}
 P_Y(T)=\frac{2^{n-1}}{n!}\prod_{v\in V_{int}(T)}\frac{1}{\kappa_T(v)-1}
\label{probY}
\end{equation}
The following easy lemma on the probability of a tree-sum under the Yule model will be used in our computations.

\begin{lemma}\label{lem:pangle}
Let $\emptyset\neq S_k\subsetneq \{1,\ldots,n\}$ with $|S_k|=k$, let  $T_k\in \TT(S_k)$ and $T'_{n-k}\in\TT( \{1,\ldots,n\}\setminus S_k)$. Then
$$
P_Y(T_k\widehat{\ }\, {}T'_{n-k})=\dfrac{2}{(n-1)\binom{n}{k}} P_Y(T_{k})P_Y(T'_{n-k})
$$
\end{lemma}

\begin{proof}
This equality is a direct consequence of applying equation (\ref{probY}) to compute $P_Y(T_k)$, $P_Y(T'_{n-k})$ and $P_Y(T_k\widehat{\ }\, {}T'_{n-k})$, using the fact that $V_{int}(T_k\widehat{\ }\, {}T'_{n-k})$ is the disjoint union of $V_{int}(T_{k})$, $V_{int}(T'_{n-k})$ and the root $r$ of $T_k\widehat{\ }\, {}T'_{n-k}$. \hspace*{\fill}\qed
\end{proof}

\subsection{Harmonic numbers}
\label{subsec:hn}

For every $n\geq 1$, let
$$
H_n=\sum_{i=1}^n \frac{1}{i},\quad H_n^{(2)}=\sum_{i=1}^n \frac{1}{i^2}.
$$
Let, moreover, $H_0=H_0^{(2)}=0$. $H_n$ is called the $n$-th \emph{harmonic number}, and $H_n^{(2)}$, the \emph{generalized harmonic number of power $2$}.
It is known (see, for instance,  \cite[p. 264]{Knuth2}) that
$$
\begin{array}{l}
\displaystyle H_n=\ln(n)+\gamma+\frac{1}{2n}-\frac{1}{12n^2}+O\Big(\frac{1}{n^3}\Big)\\
\displaystyle  H_n^{(2)}=\frac{\pi^2}{6}-\frac{1}{n}+\frac{1}{2n^2}+O\Big(\frac{1}{n^3}\Big)
\end{array}
$$
where $\gamma$ is Euler's constant.

The following identities will be used in the proofs of our main results, usually without any further notice.

\begin{lemma}\label{lem:sumharm}
For every $n\geq 2$:
\begin{enumerate}[(1)]
\item $\displaystyle \sum_{k=1}^{n-1} H_k=n(H_{n}-1)$

\item $\displaystyle \sum_{k=1}^{n-1} kH_k=\frac{1}{4}n(n-1)(2H_{n}-1)$

\item $\displaystyle \sum_{k=1}^{n-1} k^2H_k=\frac{1}{36}n(n-1)((12n-6)H_n-4n-1)$

\item $\displaystyle \sum_{k=1}^{n-1} \frac{H_k}{k+1}=\frac{1}{2}(H_n^2-H_n^{(2)})$

\item $\displaystyle  \sum_{k=1}^{n-1} H_k^2=nH_{n}^2-(2n+1)H_{n}+2n$

\item $\displaystyle  \sum_{k=1}^{n-1} H_k^{(2)}=nH_n^{(2)}-H_n$

\item $\displaystyle \sum_{k=1}^{n-1} H_kH_{n-k}=(n+1)(H_{n+1}^2-H_{n+1}^{(2)}-2H_{n+1}+2)$

\item $\displaystyle \sum_{k=1}^{n-1} kH_kH_{n-k}=\binom{n+1}{2}(H_{n+1}^2-H_{n+1}^{(2)}-2H_{n+1}+2)$

\item $ \displaystyle \sum_{k=1}^{n-1}kH_{\lfloor k/2\rfloor}=
\frac{1}{2}n (n-1) H_{\lfloor n/2\rfloor }-\Big\lfloor \frac{n}{2} \Big\rfloor^2$

\item $ \displaystyle \sum_{k=1}^{n-1} \frac{H_k}{2k-1}= \sum_{k=1}^{n-1} \frac{H_k}{2k+1} -\left( \frac{4n-1}{2n-1}\right) H_n +2 H_{2n}$
\end{enumerate}
\end{lemma}

\begin{proof}
Identities (1)--(6) are well known and easily proved by induction on $n$: see, for instance, the chapters on harmonic numbers in Knuth's classical textbooks  \cite[\S 6.3, 6.4]{Knuth2} and  \cite[\S 1.2.7]{Knuth1}. Identities (7) and (8) are proved in \cite[Thms. 1,2]{WGW}. We shall prove (9) and (10) here. 

As far as (9) goes, if $n$ is even, 
$$
\begin{array}{l}
\displaystyle \sum_{k=1}^{n-1}kH_{\lfloor k/2\rfloor}  = 
\sum_{j=1}^{(n-2)/2} 2jH_j+\sum_{j=0}^{(n-2)/2}(2j+1)H_{j}\\ \displaystyle \quad =
4\sum_{j=1}^{n/2-1} jH_j+\sum_{j=1}^{n/2-1} H_{j}=\frac{1}{2} n(n-1) H_{\frac{n}{2}}-\Big(\frac{n}{2}\Big)^2
\end{array}
$$
while if $n$ is odd, 
$$
\begin{array}{l}
\displaystyle\sum_{j=1}^{n-1}kH_{\lfloor k/2\rfloor}  = 
\sum_{j=1}^{(n-1)/2}  \hspace*{-2ex}2jH_j+ \hspace*{-2ex}\sum_{j=0}^{(n-1)/2-1} \hspace*{-2ex}(2j+1)H_{j}\\ 
\displaystyle \quad  =
4 \hspace*{-2ex}\sum_{j=1}^{(n-1)/2-1} \hspace*{-2ex} jH_j+ \hspace*{-2ex}\sum_{j=1}^{(n-1)/2-1} \hspace*{-2ex} H_{j}+(n-1)H_{\frac{n-1}{2}}
=\frac{1}{2}n (n-1) H_{\frac{n-1}{2}}-\Big(\frac{n-1}{2}\Big)^2
\end{array}
$$
Both equalities agree with identity (9). 

As far as (10) goes,
$$
\begin{array}{l}
\displaystyle \sum_{k=1}^{n-1} \frac{H_k}{2k+1} = \sum_{k=2}^n \frac{H_{k-1}}{2k-1}=\sum_{k=2}^{n} \frac{H_k-\frac{1}{k}}{2k-1} = \sum_{k=2}^n \frac{H_k}{2k-1}-\sum_{k=2}^n \frac{1}{k(2k-1)} \\ \displaystyle \qquad  =
\sum_{k=1}^{n-1} \frac{H_k}{2k-1} -1+\frac{H_n}{2n-1} +\sum_{k=2}^n \left(\frac{1}{k}-\frac{2}{2k-1}\right)  \\ \displaystyle \qquad  = 
\sum_{k=1}^{n-1} \frac{H_k}{2k-1} -1+\frac{H_n}{2n-1} + H_n -1-2\sum_{k=2}^n\frac{1}{2k-1}   \\ \displaystyle \qquad  = 
\sum_{k=1}^{n-1} \frac{H_k}{2k-1} -1 +\frac{H_n}{2n-1} +H_n -1-2(H_{2n}-\frac{1}{2} H_n)+2 \\ \displaystyle \qquad  =
\sum_{k=1}^{n-1} \frac{H_k}{2k-1}+\left( \frac{4n-1}{2n-1}\right) H_n -2 H_{2n}.
\end{array}
$$
which is equivalent to (10).\hspace*{\fill}\qed
\end{proof}

\section{Recursive shape indices}
\label{sec:rec}

A \emph{recursive shape index  for phylogenetic trees} \cite{Matsen} is a mapping $I$ that associates to each phylogenetic tree $T$ a real number $I(T)\in \mathbb{R}$ satisfying the following two conditions:
\begin{enumerate}
\item[(a)] It is  invariant under tree isomorphisms  and relabelings of leaves.

\item[(b)] There exists a symmetrical mapping $f_I:\NN\times \NN\to \RR$ such that, for every  phylogenetic trees $T,T'$ on disjoint sets of taxa $S,S'$, respectively,
$$
I(T\,\widehat{\ }\, T')=I(T)+I(T')+f_I(|S|,|S'|).
$$
\end{enumerate}

As we shall see in later sections, the balance indices considered in this paper are recursive shape indices in this sense.  The following two results extract a common part of the computation of their variances.
In them, and henceforth, $E_Y$ applied to a random variable will mean the expected value of this random variable under the Yule model.

\begin{lemma}\label{lem:YI}
Let $I$ be a recursive shape index for phylogenetic trees. 
For every $n\geq 1$, let $I_n$ be the random variable that chooses a tree $T\in \TT_n$ and computes $I(T)$. Then,
$$
\begin{array}{rl}
E_Y(I_n^2) & \displaystyle =\frac{1}{n-1}\sum_{k=1}^{n-1} \Big(2E_Y(I^2_k)  + 4f_I(k,n-k)E_Y(I_k)+2E_Y(I_k)E_Y(I_{n-k})\\
& \displaystyle \qquad\qquad \qquad\quad +f_I(k,n-k)^2\Big).
\end{array}
$$
\end{lemma}

\begin{proof}
We compute  $E_Y(I_n^2)$ using its very definition and Lemma \ref{lem:pangle}. Recall that
every tree in $\TT_n$ is obtained \emph{twice} as $T_k\widehat{\ }\, {}T'_{n-k}$, for some $1\leq k\leq n-1$, some subset  $S_k\subseteq \{1,\ldots,n\}$ with $k$ elements, some tree $T_k$ on $S_k$ and some tree $T'_{n-k}$ on $S_k^c$.
$$
\begin{array}{l}
E_Y(I_n^2) \displaystyle =\sum_{T\in \TT_n} I(T)^2\cdot P_Y(T)
\\
\quad \displaystyle =\frac{1}{2}
\sum_{k=1}^{n-1}\sum_{S_k\subsetneq\{1,\ldots,n\}\atop |S_k|=k}
 \sum_{T_k\in \TT(S_k)}\sum_{T'_{n-k}\in \TT(S_k^c)}I(T_k\widehat{\ }\, {}T'_{n-k})^2\cdot P_Y(T_k\widehat{\ }\, {}T'_{n-k})
 \\ 
\quad   \displaystyle =\frac{1}{2} \sum_{k=1}^{n-1}\binom{n}{k}
 \sum_{T_k\in \TT_k}\sum_{T'_{n-k} \in \TT_{n-k}}\big(I(T_k)+I(T'_{n-k})+f_I(k,n-k)\big)^2\\
 \quad  \displaystyle  \qquad\qquad 
\cdot
 \frac{2}{(n-1)\binom{n}{k}} P_Y(T_{k})P_Y(T'_{n-k})\\
\quad  \displaystyle  =\frac{1}{n-1}\sum_{k=1}^{n-1} 
 \sum_{T_k}\sum_{T'_{n-k}}\big[I(T_k)^2+I(T'_{n-k})^2+f_I(k,n-k)^2+2I(T_k)I(T'_{n-k})\\
 \quad  \displaystyle  \qquad\qquad 
 +2f_I(k,n-k)I(T_k)+2f_I(k,n-k)I(T'_{n-k})\big]P_Y(T_{k})P_Y(T'_{n-k})\\
\quad  \displaystyle  =\frac{1}{n-1}\sum_{k=1}^{n-1} 
\Big[ \sum_{T_k}\sum_{T'_{n-k} }I(T_k)^2P_Y(T_{k})P_Y(T'_{n-k})   \\
\quad  \displaystyle\qquad +\sum_{T_k}\sum_{T'_{n-k} }I(T'_{n-k})^2P_Y(T_{k})P_Y(T'_{n-k})\\
\quad  \displaystyle\qquad +\sum_{T_k}\sum_{T'_{n-k} }f_I(k,n-k)^2P_Y(T_{k})P_Y(T'_{n-k})\\
\quad  \displaystyle\qquad +2\sum_{T_k}\sum_{T'_{n-k} }f_I(k,n-k)I(T_k)P_Y(T_{k})P_Y(T'_{n-k})\\
\quad  \displaystyle\qquad +2\sum_{T_k}\sum_{T'_{n-k} }f_I(k,n-k)I(T'_{n-k})P_Y(T_{k})P_Y(T'_{n-k})\\
\end{array}
$$
$$
\begin{array}{l}
\quad  \displaystyle\qquad +2\sum_{T_k}\sum_{T'_{n-k} }I(T_k)I(T'_{n-k})P_Y(T_{k})P_Y(T'_{n-k}) \Big]\\
\quad  \displaystyle  =\frac{1}{n-1}\sum_{k=1}^{n-1} 
\Big[ \sum_{T_k}I(T_k)^2P_Y(T_{k})  +\sum_{T'_{n-k} }I(T'_{n-k})^2P_Y(T'_{n-k})+f_I(k,n-k)^2\\
\quad  \displaystyle\qquad +2\sum_{T_k}f_I(k,n-k)I(T_k)P_Y(T_{k})+2\sum_{T'_{n-k} }f_I(k,n-k)I(T'_{n-k})P_Y(T'_{n-k})\\
\quad  \displaystyle\qquad +2\Big(\sum_{T_k}I(T_k)P_Y(T_{k})\Big)\Big(\sum_{T'_{n-k} }I(T'_{n-k})P_Y(T'_{n-k})\Big) \Big]
  \\
\quad  \displaystyle  =\frac{1}{n-1}\sum_{k=1}^{n-1} 
\Big( E_Y(I^2_k) + E_Y(I^2_{n-k})+ f_I(k,n-k)^2\\
\quad  \displaystyle\qquad  +2f_I(k,n-k)(E_Y(I_k)+ E_Y(I_{n-k}))+2E_Y(I_k) E_Y(I_{n-k}) \Big)
 \\
 \quad  \displaystyle  =\frac{1}{n-1}\sum_{k=1}^{n-1} 
\Big(2E_Y(I_k^2) + 4f_I(k,n-k)E_Y(I_k)+2E_Y(I_k) E_Y(I_{n-k})\\ \quad  \displaystyle  \qquad\qquad \qquad \qquad 
+f_I(k,n-k)^2\Big)
\end{array}
$$
as we claimed. \hspace*{\fill}\qed
\end{proof}

\begin{corollary}\label{cor:YI}
Let $I$ be a recursive shape index for phylogenetic trees and,
for every $n\geq 1$, let $I_n$ be the random variable that chooses a tree $T\in \TT_n$ and computes $I(T)$. Set
$$
\begin{array}{l}
\varepsilon_I(a,b-1)=f_I(a,b)-f_I(a,b-1)\mbox{ for every $a\geq 1$ and $b\geq 2$}\\
R_I(n-1)=E_Y(I_n)-E_Y(I_{n-1})\mbox{ for every $n\geq  2$}
\end{array}
$$
If $E_Y(I_1)=0$, then
$$
\begin{array}{l}
\displaystyle E_Y(I_n^2)= \frac{n}{n-1}  E_Y(I_{n-1}^2)
+\frac{4}{n-1} \sum_{k=1}^{n-2} \ee_I(k,n-1-k) E_Y(I_k)\\
\quad\qquad \displaystyle+\frac{4}{n-1}f_I(n-1,1)  E_Y(I_{n-1}) +\frac{2}{n-1} \sum_{k=1}^{n-2}  E_Y(I_k)R_I(n-k-1)\\
\quad\qquad \displaystyle+\frac{f_I(n-1,1)^2}{n-1}+\frac{1}{n-1}\sum_{k=1}^{n-2} (f_I(k,n-k)^2- f_I(k,n-k-1)^2).
\end{array}
$$
\end{corollary}

\begin{proof}
By Lemma \ref{lem:YI},
$$
\begin{array}{rl}
E_Y(I_n^2)& \displaystyle =\frac{2}{n-1}\sum_{k=1}^{n-1} E_Y(I^2_k)  + \frac{4}{n-1}\sum_{k=1}^{n-1} f_I(k,n-k)E_Y(I_k)\\
& \qquad \displaystyle +\frac{2}{n-1}\sum_{k=1}^{n-1} E_Y(I_k)E_Y(I_{n-k}) +\frac{1}{n-1}\sum_{k=1}^{n-1} f_I(k,n-k)^2,
\end{array}
$$
and in particular
$$
\begin{array}{rl}
E_Y(I_{n-1}^2)& \displaystyle =\frac{2}{n-2}\sum_{k=1}^{n-2} E_Y(I^2_k)  + \frac{4}{n-2}\sum_{k=1}^{n-2} f_I(k,n-1-k)E_Y(I_k)\\
& \quad \displaystyle +\frac{2}{n-2}\sum_{k=1}^{n-2} E_Y(I_k)E_Y(I_{n-1-k}) +\frac{1}{n-2}\sum_{k=1}^{n-2} f_I(k,n-k-1)^2.
\end{array}
$$
Therefore
$$
\begin{array}{l}
\displaystyle E_Y(I_n^2)
= \frac{n-2}{n-1}\cdot \frac{2}{n-2}\sum_{k=1}^{n-2} E_Y(I^2_k) 
+\frac{2}{n-1} E_Y(I_{n-1}^2)\\
\quad\qquad \displaystyle  
+ \frac{n-2}{n-1}\cdot  \frac{4}{n-2}\sum_{k=1}^{n-2} E_Y(I_k)(f_I(k,n-1-k)+\ee_I(k,n-1-k))\\
\quad\qquad\qquad \displaystyle  
+\frac{4}{n-1}f_I(n-1,1)E_Y(I_{n-1})\\
\quad\qquad \displaystyle  
+ \frac{n-2}{n-1}\cdot \frac{2}{n-2}\sum_{k=1}^{n-2} E_Y(I_k)(E_Y(I_{n-1-k})+R_I(n-k-1))\\
\quad\qquad\qquad \displaystyle  
+ \frac{2}{n-1} E_Y(I_{n-1})E_Y(I_1)
\\
\quad\qquad \displaystyle 
+ \frac{n-2}{n-1}\cdot \frac{1}{n-2}\sum_{k=1}^{n-2} f_I(k,n-k-1)^2 
+\frac{1}{n-1}\sum_{k=1}^{n-1} f_I(k,n-k)^2\\
\quad\qquad\qquad \displaystyle-\frac{n-2}{n-1}\cdot \frac{1}{n-2}\sum_{k=1}^{n-2} f_I(k,n-k-1)^2
\\[2ex]
\quad \displaystyle 
= \frac{n-2}{n-1}  E_Y(I_{n-1}^2)+\frac{2}{n-1} E_Y(I_{n-1}^2) 
+\frac{4}{n-1} \sum_{k=1}^{n-2} \ee_I(k,n-1-k) E_Y(I_k)\\
\quad\qquad \displaystyle+\frac{4}{n-1}f_I(n-1,1) E_Y(I_{n-1}) +\frac{2}{n-1} \sum_{k=1}^{n-2} 
E_Y(I_k)R_I(n-k-1)\\
\quad\qquad \displaystyle+\frac{1}{n-1} \sum_{k=1}^{n-1} f_I(k,n-k)^2-\frac{1}{n-1}\sum_{k=1}^{n-2} f_I(k,n-k-1)^2 
\\[2ex]
\quad \displaystyle 
= \frac{n}{n-1}  E_Y(I_{n-1}^2)+\frac{4}{n-1} \sum_{k=1}^{n-2} \ee_I(k,n-1-k) E_Y(I_k)\\
\quad\qquad \displaystyle+\frac{4}{n-1}f_I(n-1,1) E_Y(I_{n-1}) +\frac{2}{n-1} \sum_{k=1}^{n-2} 
E_Y(I_k)R_I(n-k-1)\\
\quad\qquad \displaystyle+\frac{1}{n-1}\sum_{k=1}^{n-2} (f_I(k,n-k)^2- f_I(k,n-k-1)^2)+\frac{1}{n-1}\cdot f_I(n-1,1)^2
\end{array}
$$
as we claimed. \hspace*{\fill}\qed
\end{proof}

\section{The variance of Sackin's index}

The \emph{Sackin index} of a phylogenetic tree $T\in \TT_n$ is defined as the sum of the depths of its leaves:
$$
S(T)=\sum_{i=1}^n\delta_T(i).
$$
It is well known (see, for instance, \cite[Eq. (6)]{Rogers:96}) that if $T_k\in \TT(S_k)$,  for some $\emptyset\neq S_k\subsetneq \{1,\ldots,n\}$, and $T'_{n-k}\in \TT(S_k^c)$, then
$$
S(T_k\widehat{\ }\, T'_{n-k})=S(T_k)+S(T'_{n-k})+n.
$$

Let  $S_n$ be the random variable that chooses a tree $T\in\TT_n$ and computes $S(T)$.  Its expected value  under the Yule model
is  \cite{Heard92}
$$
E_Y(S_n)=2n(H_n-1).
$$
In particular $E_Y(S_1)=0$. Actually, the Sackin index of a tree with only one node is 0.
Notice moreover  that $E_Y(S_n)$ satisfies the recurrence
$$
E_Y(S_{n+1})=E_Y(S_n)+2H_n.
$$
Indeed,
$$
\begin{array}{rl}
E_Y(S_{n+1})-E_Y(S_n) & =2(n+1)(H_{n+1}-1)-2n(H_n-1)\\ & \displaystyle =
2(n+1)(H_{n}+\frac{1}{n+1}-1)-2n(H_n-1)=2H_n.
\end{array}
$$

In this section we prove that the variance of $S_n$ under this model is (see Cor. \ref{th:varS})
$$
\sigma^2_Y(S_n)=7n^2-4n^2H_n^{(2)}-2nH_n-n.
$$

\begin{theorem}\label{th:S2}
$\displaystyle E_Y(S_n^2)= 4n^2(H_n^2-H_n^{(2)} -2H_n) -2nH_n+11n^2-n$
\end{theorem}

\begin{proof}
As we have seen, Sackin's index satisfies the hypotheses in Corollary \ref{cor:YI}, with $f_S(k,n-k)=n$, and hence $\ee_S(k,n-k-1)=1$,  and
$R_S(k)=2H_{k}$. Therefore
$$
\begin{array}{l}
\displaystyle E_Y(S_n^2)= \frac{n}{n-1}  E_Y(S_{n-1}^2)+\frac{4}{n-1} \sum_{k=1}^{n-2} E_Y(S_k)\\
\quad\qquad \displaystyle+\frac{4}{n-1}n E_Y(S_{n-1}) +\frac{2}{n-1} \sum_{k=1}^{n-2} 
E_Y(S_k)2H_{n-k-1}\\
\quad\qquad \displaystyle+\frac{n^2}{n-1}+\frac{1}{n-1} \sum_{k=1}^{n-2} (n^2 - (n-1)^2)
\\[2ex]
\quad \displaystyle = \frac{n}{n-1}  E_Y(S_{n-1}^2)+\frac{4}{n-1} \sum_{k=1}^{n-2} E_Y(S_k)+8n(H_{n-1}-1)\\
\quad\qquad \displaystyle +\frac{4}{n-1} \sum_{k=1}^{n-2} 
E_Y(S_k)H_{n-k-1}+3n-2
\end{array}
$$
Now, by Lemma \ref{lem:sumharm},
$$
\begin{array}{l}
\displaystyle \frac{4}{n-1} \sum_{k=1}^{n-2} E_Y(S_k)=
\frac{8}{n-1}\sum_{k=1}^{n-2} k(H_k-1)\\
\qquad \displaystyle=\frac{8}{n-1}\Big(\frac{1}{4}(n-1)(n-2)(2H_{n-1}-1)-\frac{1}{2}(n-1)(n-2)\Big)\\
\qquad \displaystyle =2(n-2)(2H_{n-1}-3)
\end{array}
$$
and
$$
\begin{array}{l}
\displaystyle \frac{4}{n-1} \sum_{k=1}^{n-2} 
E_Y(S_k)H_{n-k-1}=\frac{8}{n-1}\sum_{k=1}^{n-2} 
k(H_k-1)H_{n-k-1}
\\
\qquad \displaystyle =\frac{8}{n-1}\sum_{k=1}^{n-2} 
kH_kH_{n-k-1}-\frac{8}{n-1}\sum_{k=1}^{n-2} 
kH_{n-k-1}\\
\qquad \displaystyle =\frac{8}{n-1}\sum_{k=1}^{n-2} 
kH_kH_{n-k-1}-\frac{8}{n-1}\sum_{k=1}^{n-2} 
(n-k-1)H_{k}\\
\qquad \displaystyle =
4n(H_{n}^2-H_{n}^{(2)}-2H_{n}+2)-8\sum_{k=1}^{n-2} 
H_{k}+\frac{8}{n-1}\sum_{k=1}^{n-2} kH_k\\
\qquad \displaystyle =
4n(H_{n}^2-H_{n}^{(2)}-2H_{n}+2)-8(n-1)(H_{n-1}-1)\\
\qquad\qquad \displaystyle +2(n-2)(2H_{n-1}-1)
\\
\qquad \displaystyle =
4n(H_{n}^2-H_{n}^{(2)}-2H_{n-1}-2\cdot \frac{1}{n}+2)-4nH_{n-1}+6n-4
\\
\qquad \displaystyle =
4n(H_{n}^2-H_{n}^{(2)}-3H_{n-1})+14n-12
\end{array}
$$
and thus
$$
\begin{array}{rl}
\displaystyle E_Y(S_n^2) & \displaystyle =\frac{n}{n-1}  E_Y(S_{n-1}^2)+2(n-2)(2H_{n-1}-3)+8n(H_{n-1}-1)\\
& \qquad\qquad \displaystyle +4n(H_{n}^2-H_{n}^{(2)}-3H_{n-1})+14n-12+3n-2\\
&\displaystyle =\frac{n}{n-1}  E_Y(S_{n-1}^2)+4n(H_n^2-H_n^{(2)})-8H_{n-1}+3n-2
\end{array}
$$
Setting $x_n=E_Y(S_n^2)/n$, this equation becomes
$$
x_n=x_{n-1}+4(H_n^2-H_n^{(2)})-8\frac{H_{n-1}}{n}+3-\frac{2}{n}
$$
The solution of this  recursive equation with $x_1=0$ is
$$
\begin{array}{rl}
x_n & \displaystyle =\sum_{k=2}^n\Big(4(H_k^2-H_k^{(2)})-8\frac{H_{k-1}}{k}+3-\frac{2}{k}\Big)\\
&  \displaystyle =4\sum_{k=2}^n (H_k^2-H_k^{(2)})-8\sum_{k=1}^{n-1}\frac{H_{k}}{k+1}+3(n-1)-2\sum_{k=2}^{n}\frac{1}{k}\\
&  \displaystyle =4\sum_{k=2}^n (H_k^2-H_k^{(2)})-4(H_n^2-H_n^{(2)})+3(n-1)-2(H_n-1)\\
&  \displaystyle =4\sum_{k=2}^{n-1} (H_k^2-H_k^{(2)})-2H_n+3n-1\\
& \displaystyle =4(nH_n^2-(2n+1)H_n+2n-nH_n^{(2)}+H_n)-2H_n+3n-1
\\
& \displaystyle =4n(H_n^2-H_n^{(2)} -2H_n)-2H_n +11n-1
\end{array}
$$
and therefore
$$
E_Y(S_n^2)=nx_n=4n^2(H_n^2-H_n^{(2)} -2H_n) -2nH_n +11n^2-n
$$
as we claimed. \hspace*{\fill}\qed
\end{proof}

\begin{corollary}\label{th:varS}
The variance of $S_n$ under the Yule model is
$$
\sigma^2_Y(S_n)=7n^2-4n^2H_n^{(2)}-2nH_n-n.
$$
\end{corollary}

\begin{proof}
This formula is obtained by replacing
$$
\begin{array}{rl}
E_Y(S_n^2) & =4n^2(H_n^2-H_n^{(2)} -2H_n) -2nH_n +11n^2-n\\
E_Y(S_n) & =2n(H_n-1)
\end{array}
$$
 in the identity $\sigma^2_Y(S_n)=E_Y(S_n^2)-E_Y(S_n)^2$. \hspace*{\fill}\qed
\end{proof}

From this exact formula we can obtain an $O(1/n)$ approximation of $\sigma^2_Y(S_n)$, which refines the limit formula obtained in  \cite{BFJ:06}.

\begin{corollary}
$\displaystyle \sigma^2_Y(S_n)=\Big(7-\frac{2\pi^2}{3} \Big)n^2+n(3-2\ln(n)-2\gamma)-3+O\Big(\frac{1}{n}\Big).$
\end{corollary}

\section{The variance of the total cophenetic index $\Phi$}

The \emph{total cophenetic index} of a phylogenetic tree $T\in \TT_n$ is defined as the sum of the cophenetic values of its pairs of leaves:
$$
\Phi(T)=\sum_{1\leq i<j\leq n}\varphi_T(i,j).
$$
By \cite[Lem. 4]{MRR},  if $T_k\in \TT(S_k)$,  for some $\emptyset\neq S_k\subsetneq \{1,\ldots,n\}$ with $k$ elements, and $T'_{n-k}\in \TT(S_k^c)$, then
$$
\Phi(T_k\widehat{\ }\, {}T_{n-k})= \Phi(T_k)+ \Phi(T_{n-k})+\binom{k}{2}+\binom{n-k}{2}.
$$
Therefore, $\Phi$ is a recursive shape index with $f_{\Phi}(k,n-k)=\binom{k}{2}+\binom{n-k}{2}$, and in particular $\ee_{\Phi}(k,n-k-1)=n-k-1$.

Let  $\Phi_n$ be the random variable that chooses a tree $T\in\TT_n$ and computes its total cophenetic index $\Phi(T)$.
The expected value  under the Yule model of $\Phi_n$ 
is  \cite{MRR}
$$
E_Y(\Phi_n)=n(n-1)-2n(H_n-1)=n(n+1-2H_n).
$$
In particular, $E_Y(\Phi_1)=0$.  Actually, the total cophenetic index of a tree with only one node is 0. Moreover, we have that
$$
E_Y(\Phi_n)=E_Y(\Phi_{n-1})+2(n-1-H_{n-1}),
$$
and therefore $R(k)=2(k-H_{k})$.

In this section we prove that the variance of $\Phi_n$ under this model is (see Cor. \ref{th:varphi})
$$
\sigma^2_Y(\Phi_n)=\frac{1}{12}(n^4-10n^3+131 n^2-2n)-4n^2 H_n^{(2)}-6nH_n
$$

\begin{theorem}\label{th:phi2}
$$
\begin{array}{rl}
E_Y(\Phi_n^2) & = \displaystyle \frac{1}{12}(13n^4+14n^3+143 n^2-2n)+4n^2(H_n^2-H_n^{(2)})\\[2ex]
& \quad -2(2n^3+2n^2+3n)H_n
\end{array}
$$
\end{theorem}

\begin{proof}
As we have seen, $\Phi$ satisfies the hypotheses of Corollary \ref{cor:YI}, with
$$
\ee_\Phi(k,n-k-1)=n-k-1,\quad R(k)=2(k-H_k).
$$
Therefore, by the aforementioned result,
$$
\begin{array}{l}
\displaystyle E_Y(\Phi_n^2)=  \frac{n}{n-1}  E_Y(\Phi_{n-1}^2)+\frac{4}{n-1} \sum_{k=1}^{n-2} (n-k-1)E_Y(\Phi_k)\\
\quad\qquad \displaystyle+\frac{4}{n-1}\binom{n-1}{2} E_Y(\Phi_{n-1})+\frac{1}{n-1}\binom{n-1}{2}^2\\
\quad\qquad \displaystyle +\frac{2}{n-1} \sum_{k=1}^{n-2} 2
E_Y(\Phi_k)((n-k-1)-H_{n-k-1})\\
\quad\qquad \displaystyle+\frac{1}{n-1} \sum_{k=1}^{n-2} \Big(\Big(\binom{k}{2}+\binom{n-k}{2}\Big)^2- \Big(\binom{k}{2}+\binom{n-k-1}{2}\Big)^2 \Big)\\[2ex]
 \quad \displaystyle =
 \frac{n}{n-1}E_Y(\Phi_{n-1}^2)+\frac{8}{n-1}\sum_{k=1}^{n-2}(n-k-1)(k^2+k-2kH_k)\\
\quad\qquad \displaystyle+2(n-2)(n-1)(n-2H_{n-1})\\
 \qquad\quad \displaystyle -\frac{4}{n-1}\sum_{k=1}^{n-2}H_{n-k-1}(k^2+k-2kH_k)+\frac{1}{12}(n-2)(7n^2-21n+12)
\\[2ex]
\quad \displaystyle =
 \frac{n}{n-1}E_Y(\Phi_{n-1}^2)-4(n-2)(n-1)H_{n-1}\\
 \qquad\quad \displaystyle
 -16\sum_{k=1}^{n-2}kH_k+\frac{16}{n-1}\sum_{k=1}^{n-2}k^2H_k\\
 \qquad\quad \displaystyle
 -\frac{4}{n-1}\sum_{k=1}^{n-2}k^2H_{n-k-1}- \frac{4}{n-1}\sum_{k=1}^{n-2}kH_{n-k-1}
\\
 \qquad\quad \displaystyle +\frac{8}{n-1}\sum_{k=1}^{n-2}kH_kH_{n-k-1}+\frac{1}{12}(n-2)(39n^2-37n+12)
 \\[2ex]
  \quad \displaystyle =
 \frac{n}{n-1}E_Y(\Phi_{n-1}^2)-4(n-2)(n-1)H_{n-1}\\
 \qquad\quad \displaystyle
 -16\sum_{k=1}^{n-2}kH_k+\frac{16}{n-1}\sum_{k=1}^{n-2}k^2H_k\\
 \qquad\quad \displaystyle
 -\frac{4}{n-1}\sum_{k=1}^{n-2}(n-k-1)^2H_{k}- \frac{4}{n-1}\sum_{k=1}^{n-2}(n-k-1)H_{k}
\\
 \qquad\quad \displaystyle +\frac{8}{n-1}\sum_{k=1}^{n-2}kH_kH_{n-k-1}+\frac{1}{12}(n-2)(39n^2-37n+12)
 \\[2ex] 
 \end{array}
$$
$$
\begin{array}{l}
  \quad \displaystyle =
 \frac{n}{n-1}E_Y(\Phi_{n-1}^2)-4(n-2)(n-1)H_{n-1}\\
 \qquad\quad \displaystyle
-16\sum_{k=1}^{n-2}kH_k+\frac{16}{n-1}\sum_{k=1}^{n-2}k^2H_k+\frac{8}{n-1}\sum_{k=1}^{n-2}kH_kH_{n-k-1}\\
 \qquad\quad \displaystyle-4(n-1)\sum_{k=1}^{n-2} H_k+8\sum_{k=1}^{n-2}kH_k-\frac{4}{n-1}\sum_{k=1}^{n-2}k^2H_k-4\sum_{k=1}^{n-2}H_k
\\
 \qquad\quad \displaystyle +\frac{4}{n-1}\sum_{k=1}^{n-2}kH_k+\frac{1}{12}(n-2)(39n^2-37n+12)
 \\[2ex]
  \quad \displaystyle =
 \frac{n}{n-1}E_Y(\Phi_{n-1}^2)-4(n-2)(n-1)H_{n-1}\\
 \qquad\quad \displaystyle
+\frac{12}{n-1}\sum_{k=1}^{n-2}k^2H_k+\frac{12-8n}{n-1}\sum_{k=1}^{n-2}kH_k-4n\sum_{k=1}^{n-2}H_k
\\
 \qquad\quad \displaystyle+\frac{8}{n-1}\sum_{k=1}^{n-2}kH_kH_{n-k-1} +\frac{1}{12}(n-2)(39n^2-37n+12)\\[2ex]
 \quad \displaystyle =
 \frac{n}{n-1}E_Y(\Phi_{n-1}^2)-4(n-2)(n-1)H_{n-1}\\
 \qquad\quad \displaystyle
+\frac{12}{n-1}\cdot \frac{1}{36}(n-1)(n-2)\big((12n-18)H_{n-1}-4n+3\big)\\
 \qquad\quad \displaystyle
+\frac{12-8n}{n-1}\cdot \frac{1}{4}(n-1)(n-2)(2H_{n-1}-1)-4n(n-1)(H_{n-1}-1)\\
 \qquad\quad \displaystyle
+\frac{8}{n-1}\binom{n}{2}(H_n^2-H_n^{(2)}-2H_n+2)+\frac{1}{12}(n-2)(39n^2-37n+12)
\\[2ex]
 \quad \displaystyle =
 \frac{n}{n-1}E_Y(\Phi_{n-1}^2)+4n(H_n^2-H_n^{(2)})-8nH_n\\
 \qquad\quad \displaystyle
-8(n-1)^2H_{n-1}+\frac{1}{12}(39 n^3-59 n^2+94 n+24)\\[2ex]
 \quad \displaystyle =
\frac{n}{n-1}E_Y(\Phi_{n-1}^2)+4n(H_n^2-H_n^{(2)})
-8(n^2-n+1)H_{n-1}\\
\qquad \quad  \displaystyle
+\frac{1}{12}(39 n^3-59 n^2+94 n-72)
 \end{array}
$$
Setting $x_n=E_Y(\Phi_n^2)/n$, this equation becomes
$$
x_n=x_{n-1}+4(H_n^2-H_n^{(2)})
-8\Big(n-1+\frac{1}{n}\Big)H_{n-1}
+\frac{1}{12}\Big(39 n^2-59 n+94-\frac{72}{n}\Big)
$$
The solution of this recursive equation with $x_1=0$ is
$$
\begin{array}{rl}
x_n & \displaystyle =\sum_{k=2}^n\Big(4(H_k^2-H_k^{(2)})
-8 \Big(k-1+\frac{1}{k} \Big)H_{k-1}
+\frac{1}{12}\Big(39 k^2-59 k+94-\frac{72}{k}\Big)\Big)\\
& \displaystyle =4\sum_{k=2}^n H_k^2-4\sum_{k=2}^nH_k^{(2)}
-8\sum_{k=1}^{n-1} kH_{k} -8\sum_{k=1}^{n-1} \frac{H_{k}}{k+1}\\
& \qquad  \displaystyle
+\frac{1}{12}\sum_{k=2}^n\Big(39 k^2-59 k+94-\frac{72}{k}\Big)
\\
& \displaystyle =4\sum_{k=2}^{n-1} H_k^2-4\sum_{k=2}^{n-1}H_k^{(2)}
-8\sum_{k=1}^{n-1} kH_{k}+\frac{1}{12}\sum_{k=2}^n\Big(39 k^2-59 k+94-\frac{72}{k}\Big)
\end{array}
$$
$$
\begin{array}{rl}
\hphantom{x_n}  &\displaystyle = 4n(H_n^2-H_n^{(2)})-8nH_n+8n-2n(n-1)(2H_{n}-1)
-6H_n\\
& \qquad  \displaystyle
+\frac{1}{12}(13 n^3-10 n^2+71 n-2)\\
& \displaystyle = 4n(H_n^2-H_n^{(2)})-2(2n^2+2n+3)H_n+\frac{1}{12}(13n^3+14n^2+143 n-2)\\
\end{array}
$$
Therefore
$$
\begin{array}{rl}
E_Y(\Phi_n^2) & \displaystyle =nx_n=4n^2(H_n^2-H_n^{(2)})-2(2n^3+2n^2+3n)H_n\\ & \qquad \displaystyle +\frac{1}{12}(13n^4+14n^3+143 n^2-2n)
\end{array}$$
as we claimed. \hspace*{\fill}\qed
\end{proof}

\begin{corollary}\label{th:varphi}
The covariance of $\Phi_n$ under the Yule model is
$$
\sigma^2_Y(\Phi_n)=\frac{1}{12}(n^4-10n^3+131 n^2-2n)-4n^2 H_n^{(2)}-6nH_n
$$
\end{corollary}

\begin{proof}
Simply replace in the formula $\sigma^2_Y(\Phi_n)=E_Y(\Phi_n^2)-E_Y(\Phi_n)^2$ the value of $E_Y(\Phi_n^2)$ obtained in the last theorem and the value of $E_Y(\Phi_n)$ recalled above. \hspace*{\fill}\qed
\end{proof}

\begin{corollary}
$$
\begin{array}{rl}
\sigma^2_Y(\Phi_n) & =\displaystyle  \frac{1}{12}n^4-\frac{5}{6}n^3+\Big(\frac{131}{12}-\frac{2\pi^2}{3}\Big)n^2-6n\ln(n) +\Big(\frac{23}{6}-6\gamma)n-5\\
&\quad \displaystyle +O\Big(\frac{1}{n}\Big)
\end{array}
$$\end{corollary}

\section{The variance of $\oF$}

For every $T\in \TT_n$, let
$$
\oF(T)=S(T)+\Phi(T)=\sum\limits_{1\leq i\leq j\leq n} \varphi_T(i,j)
$$

\begin{lemma}\label{lem:hat}
If $T_k\in \TT(S_k)$, with $\emptyset\neq S_k\subsetneq \{1,\ldots,n\}$ and $|S_k|=k$, and $T'_{n-k}\in\TT(S_k^c)$, then
$$
\oF(T_k\widehat{\ }\, {}T'_{n-k})=\oF(T_k)+\oF(T'_{n-k})+\binom{k+1}{2}+\binom{n-k+1}{2}.
$$
\end{lemma}

\begin{proof}
Since $S(T_k\widehat{\ }\, {}T'_{n-k})=S(T_k)+S(T'_{n-k})+n$ and 
$$
\Phi(T_k\widehat{\ }\, {}T'_{n-k})= \Phi(T_k)+ \Phi(T'_{n-k})+\binom{k}{2}+\binom{n-k}{2},
$$
 we have that
$$
\begin{array}{rl}
\oF(T_k\widehat{\ }\, {}T'_{n-k}) & \displaystyle =\oF(T_k)+\oF(T'_{n-k})+\binom{k}{2}+\binom{n-k}{2}+n\\ & 
\displaystyle =\oF(T_k)+\oF(T'_{n-k})+\binom{k+1}{2}+\binom{n-k+1}{2}.
\end{array}
$$
{}\ \hspace*{\fill}\qed
\end{proof}

So, $\oF$ is a recursive shape index for phylogenetic trees with
$f_{\oF}(a,b)=\binom{a+1}{2}+\binom{b+1}{2}$,
and hence $\ee_{\oF}(a,b)=b+1$.

Let  $\oF_n$ be the random variable that chooses a tree $T\in\TT_n$ and computes $\oF(T)$.  Its expected value  under the Yule model
is \cite{MRR} 
$$
E_Y(\oF_n)=n(n-1).
$$
In particular, $E_Y(\oF_1)=0$ (actually, $\oF_1=0$) and $R_{\oF}(k)=2k$. In this section we compute the variance of $\oF_n$.

In this section we prove that the variance of $\oF_n$ under this model is (see Cor. \ref{th:varoF})
$$
\sigma^2_Y(\oF_n)=2\binom{n}{4}
$$

\begin{theorem}\label{th:oF}
$E_Y(\oF_n^2)=\frac{1}{12}(13n^4-30n^3+23n^2-6n).$
\end{theorem}

\begin{proof}
$\oF$ is a recursive shape index for phylogenetic trees with
$$
\ee_{\oF}(k,n-k-1)=n-k,\quad
R(k)=2k.
$$
Then, by Corollary \ref{cor:YI},
$$
\begin{array}{l}
\displaystyle  E_Y(\oF_n^2)= \frac{n}{n-1}  E_Y(\oF_{n-1}^2)+\frac{4}{n-1} \sum_{k=1}^{n-2} (n-k)E_Y(\oF_k)+\frac{1}{n-1}\Big(\binom{n}{2}+1\Big)^2\\
\qquad \displaystyle+\frac{4}{n-1}\Big(\binom{n}{2}+1\Big) E_Y(\oF_{n-1}) +\frac{2}{n-1} \sum_{k=1}^{n-2} 2
E_Y(\oF_k)(n-k-1)\\
\qquad \displaystyle+\frac{1}{n-1} \sum_{k=1}^{n-2} \Big(\Big(\binom{k+1}{2}+\binom{n-k+1}{2}\Big)^2- \Big(\binom{k+1}{2}+\binom{n-k}{2}\Big)^2 \Big)
\\
\quad \displaystyle=\frac{n}{n-1} E_Y(\oF_{n-1}^2)\\
\qquad \displaystyle +
\frac{4}{n-1}\sum_{k=1}^{n-2} (n-k)k(k-1)+\frac{4}{n-1}\Big(\binom{n}{2}+1\Big)(n-1)(n-2)\\
\qquad \displaystyle +
\frac{4}{n-1} \sum_{k=1}^{n-2}k(k-1)(n-k-1)+
\frac{1}{n-1}\Big(\binom{n}{2}+1\Big)^2\\
\qquad \displaystyle +
\frac{1}{n-1}\sum_{k=1}^{n-2}(2(n-k)\Big(\binom{k+1}{2}+\binom{n-k}{2}\Big)+(n-k)^2)\\
\quad \displaystyle=\frac{n}{n-1} E_Y(\oF_{n-1}^2)+\frac{1}{4}n(13n^2-33n+22)
\end{array}
$$
Setting $x_n=E_Y(\oF_n^2)/n$, this recurrence becomes
$$
x_n=x_{n-1}+\frac{1}{4}(13n^2-33n+22)
$$
and the solution of this recursive equation with $x_1=0$ is
$$
x_n=\frac{1}{12}(13n^3-30n^2+23n-6)
$$
from where we deduce that
$$
E_Y(\oF_n^2)=nx_n=\frac{1}{12}(13n^4-30n^3+23n^2-6n)
$$
as we claimed. \hspace*{\fill}\qed
\end{proof}

\begin{corollary}\label{th:varoF}
The variance of $\oF_n$ under the Yule model is
$$
\sigma^2_Y(\oF_n)=2\binom{n}{4}
$$
\end{corollary}

\begin{proof}
Simply apply that $\sigma^2_Y(\oF_n)=E_Y(\oF_n^2)-E_Y(\oF_n)^2$.   \hspace*{\fill}\qed
\end{proof}

\section{The variance of Colless' index}

The \emph{Colless index} $C(T)$ of a phylogenetic tree $T\in \TT_n$ is defined as the sum, over all its internal nodes $v$,  of the absolute value of the difference between the number of descendant leaves of $v$'s children. In other words, if for every internal node $v$ we let $v_1,v_2$ denote its children, then 
$$
C(T)=\sum_{v\in V_{int}(T)} |\kappa_T(v_1)-\kappa_T(v_2)|
$$
It is well known (see, for instance, \cite[p. 100]{Rogers:96}) that if $T_k\in \TT(S_k)$,  for some $\emptyset\neq S_k\subsetneq \{1,\ldots,n\}$ with $k$ elements, and $T'_{n-k}\in \TT(\{1,\ldots,n\}\setminus S_k)$, then
$$
C(T_k\widehat{\ }\, T'_{n-k})=C(T_k)+C(T'_{n-k})+|n-2k|.
$$
In particular, it is a recursive shape index with
$$
f_C(a,b)=|b-a|.
$$
We have, then, 
$$
\ee_C(a,b-1)=f_C(a,b)-f_C(a,b-1)=|b-a|-|b-1-a|=\left\{\begin{array}{ll}
1 & \mbox{if $b\geq a+1$}\\
-1 & \mbox{if $b\leq a$}
\end{array}\right.
$$
Let  $C_n$ be the random variable that chooses a tree $T\in\TT_n$ and computes $C(T)$.  Its expected value  under the Yule model
is  \cite{Heard92}
$$
E_Y(C_n)=n(H_{\lfloor \frac{n}{2}\rfloor}-1)+(n \mmod 2)=nH_{\lfloor \frac{n}{2}\rfloor}-2\Big\lfloor \frac{n}{2}\Big\rfloor.
$$
In particular, $E_Y(C_1)=0$ and $E_Y(C_n)$ satisfies the recurrence
$$
E_Y(C_{n+1})=E_Y(C_n)+ H_{\lfloor \frac{n}{2}\rfloor}.
$$
Indeed 
$$
\begin{array}{rl}
E_Y(C_{n+1})-E_Y(C_n)  & =
(n+1)(H_{\lfloor \frac{n+1}{2}\rfloor}-1)+((n+1) \mmod 2)\\ & \qquad -
n(H_{\lfloor \frac{n}{2}\rfloor}-1)-(n \mmod 2) =(*)
\end{array}
$$
Now we distinguish two cases, depending on the parity of $n$
\begin{itemize}
\item If $n$ is even
$$
(*) = (n+1)(H_{\frac{n}{2}}-1)+1-n(H_{\frac{n}{2}}-1)=H_{\frac{n}{2}}=H_{\lfloor \frac{n}{2}\rfloor}
$$

\item If $n$ is odd
$$
\begin{array}{rl}
(*) & = (n+1)(H_{\frac{n+1}{2}}-1) -n(H_{\frac{n-1}{2}}-1)-1\\ & \displaystyle =
(n+1)\Big(H_{\frac{n-1}{2}}+\frac{2}{n+1}\Big)-nH_{\frac{n-1}{2}}-2=
H_{\frac{n-1}{2}}=H_{\lfloor \frac{n}{2}\rfloor}
\end{array}
$$
\end{itemize}

In this section we prove that the variance of $C_n$ under this model is (see Cor. \ref{th:varC})
$$
  \begin{array}{l}
\sigma_Y^2(C_n^2)  \displaystyle  =\frac{5n^2+7n}{2}+(6n+1)\Big\lfloor\frac{n}{2}\Big\rfloor
-4\Big\lfloor\frac{n}{2}\Big\rfloor^2+8\Big\lfloor\frac{n+2}{4}\Big\rfloor^2\\
\qquad\displaystyle -8(n+1)\Big\lfloor\frac{n+2}{4}\Big\rfloor
-6nH_n+\Big(2\Big \lfloor\frac{n}{2}\Big\rfloor - n(n-3)\Big)H_{\lfloor\frac{n}{2}\rfloor}\\
\qquad\displaystyle -n^2H_{\lfloor\frac{n}{2}\rfloor}^{(2)}+\Big(n^2+3n-2\Big\lfloor\frac{n}{2}\Big\rfloor\Big)H_{\lfloor\frac{n+2}{4}\rfloor}-2nH_{\lfloor\frac{n}{4}\rfloor}
\end{array}
$$

\begin{theorem}\label{th:S}
 $$
 \displaystyle \begin{array}{l}
E_Y(C_n^2)\displaystyle  =\frac{5n^2+7n}{2}+(6n+1)\Big\lfloor\frac{n}{2}\Big\rfloor
-4\Big\lfloor\frac{n}{2}\Big\rfloor^2+8\Big\lfloor\frac{n+2}{4}\Big\rfloor^2\\
\qquad\displaystyle +n^2(H_{\lfloor\frac{n}{2}\rfloor}^2-H_{\lfloor\frac{n}{2}\rfloor}^{(2)})-6nH_n+\Big(3n-n^2-(4n-2)\Big\lfloor\frac{n}{2}\Big\rfloor\Big)H_{\lfloor\frac{n}{2}\rfloor}\\
\qquad\displaystyle +\Big(n^2+3n-2\Big\lfloor\frac{n}{2}\Big\rfloor\Big)H_{\lfloor\frac{n+2}{4}\rfloor}-2nH_{\lfloor\frac{n}{4}\rfloor}
\end{array}
$$
\end{theorem}

\begin{proof}
As we have mentioned, Colless' index satisfies the hypotheses in Corollary \ref{cor:YI} with 
$$
\begin{array}{c}
R_C(n)=H_{\lfloor \frac{n}{2}\rfloor},\ f_C(k,n-k)=|n-2k|, \\[2ex] \ee_C(k,n-k-1)=\left\{\begin{array}{ll}
1 & \mbox{ if $2k\leq n-1$}\\
-1 & \mbox{ if $2k>n-1$}
\end{array}\right.
\end{array}
$$
Therefore, by the aforementioned lemma, for $n\geq 2$
$$
\begin{array}{l}
\displaystyle E_Y(C_n^2)
= \frac{n}{n-1}  E_Y(C_{n-1}^2)
+\frac{4}{n-1} \sum_{k=1}^{n-2} \ee_C(k,n-1-k) E_Y(C_k)\\
\quad\qquad \displaystyle+\frac{4}{n-1} (n-2)E_Y(C_{n-1}) +\frac{2}{n-1} \sum_{k=1}^{n-2} 
E_Y(C_k)H_{\lfloor \frac{n-k-1}{2}\rfloor}\\
\quad\qquad \displaystyle +\frac{(n-2)^2}{n-1}+\frac{1}{n-1} \sum_{k=1}^{n-2} \Big((n-2k)^2-(n-2k-1)^2 \Big)
\\[2ex]
\quad \displaystyle 
= 
\frac{n}{n-1}  E_Y(C_{n-1}^2)+\frac{4}{n-1} \sum_{k=1}^{n-2} \ee_C(k,n-1-k) E_Y(C_k)\\
\quad\qquad \displaystyle +4(n-2)H_{\lfloor \frac{n-1}{2}\rfloor}+ \frac{4(n-2)}{n-1}((n-1) \mmod 2)\\
\quad\qquad \displaystyle
+\frac{2}{n-1} \sum_{k=1}^{n-2} 
E_Y(C_k)H_{\lfloor \frac{n-k-1}{2}\rfloor} -3(n-2)
\end{array}
$$
We need to compute now 
$$
X_n= \frac{4}{n-1} \sum_{k=1}^{n-2} \ee_C(k,n-1-k) E_Y(C_k),\
Y_n=\frac{2}{n-1} \sum_{k=1}^{n-2} E_Y(C_k)H_{\lfloor \frac{n-k-1}{2}\rfloor}
$$
 In both of them, we shall have to distinguish several cases, depending on the parity of $n$.

\begin{enumerate}[(1)]
\item Set $X_n=\displaystyle \frac{4}{n-1} \sum_{k=1}^{n-2} \ee_C(k,n-1-k) E_Y(C_k)$. Notice that
$$
\begin{array}{l}
\ee_C(k,n-k-1)=\left\{\begin{array}{ll}
1 & \mbox{ if $k\leq (n-1)/2$}\\
-1 & \mbox{ if $k> (n-1)/2$}
\end{array}\right.\\[2ex]
\ee_C(k,n-k-2)=\left\{\begin{array}{ll}
1 & \mbox{ if $k\leq (n-2)/2$}\\
-1 & \mbox{ if $k> (n-2)/2$}
\end{array}\right.
\end{array}
$$

Now, on the one hand, if $n-1$ is even,
$$
\ee_C(k,n-k-1)-\ee_C(k,n-k-2)=\left\{\begin{array}{ll}
0 & \mbox{ if $k< (n-1)/2$}\\
2 & \mbox{ if $k= (n-1)/2$}\\
0 & \mbox{ if $k> (n-1)/2$}
\end{array}\right.
$$
Then, 
$$
\begin{array}{rl}
X_n & \displaystyle =\frac{n-2}{n-1}\cdot \frac{4}{n-2} \sum_{k=1}^{n-3} \ee_C(k,n-k-1) E_Y(C_k)\\ 
& \displaystyle\qquad\qquad + \frac{4}{n-1} \ee_C(n-2,1)E_Y(C_{n-2})\\
 & \displaystyle=
\frac{n-2}{n-1} \cdot\frac{4}{n-2} \sum_{k=1}^{n-3} \ee_C(k,n-k-2) E_Y(C_k)+
\frac{8}{n-1}E_Y(C_{\frac{n-1}{2}}) \\ 
& \displaystyle\qquad-\frac{4}{n-1}E_Y(C_{n-2})\\
 & \displaystyle=
\frac{n-2}{n-1}X_{n-1} +
\frac{8}{n-1}\Big(\frac{n-1}{2}(H_{\lfloor \frac{n-1}{4}\rfloor}-1)+\Big(\frac{n-1}{2} \mmod 2\Big)\Big)\\
 & \qquad\qquad \displaystyle -\frac{4}{n-1}\big((n-2)(H_{\lfloor \frac{n-2}{2}\rfloor}-1)+((n-2)\mmod 2)\big)\\
 & \displaystyle=
\frac{n-2}{n-1}X_{n-1} +4H_{\lfloor \frac{n-1}{4}\rfloor}-\frac{4(n-2)}{n-1}H_{ \frac{n-3}{2}}\\
 & \qquad\qquad \displaystyle+
\frac{8}{n-1}\Big(\Big(\frac{n-1}{2} \mmod 2\Big)-1\Big)
\end{array}
$$

On the other hand, if $n-1$ is odd, then
$$
\ee_C(k,n-k-1)=\ee_C(k,n-k-2)\quad \mbox{for every $k=1,\ldots,n-3$}
$$
and therefore
$$
\begin{array}{rl}
X_n & \displaystyle =\frac{n-2}{n-1}\cdot \frac{4}{n-2} \sum_{k=1}^{n-3} \ee_C(k,n-k-1)\cdot E_Y(C_k)\\
 & \qquad\qquad \displaystyle +
\frac{4}{n-1} \ee_C(n-2,1)E_Y(C_{n-2})\\
 & \displaystyle=
\frac{n-2}{n-1} \cdot\frac{4}{n-2} \sum_{k=1}^{n-3} \ee_C(k,n-k-2)\cdot E_Y(C_k)\\
 & \qquad\qquad \displaystyle - \frac{4}{n-1} \big((n-2)(H_{\lfloor \frac{n-2}{2}\rfloor}-1)+((n-2)\mmod 2)\big)\\
 & \displaystyle=
\frac{n-2}{n-1}X_{n-1} -\frac{4(n-2)}{n-1} (H_{ \frac{n-2}{2}}-1)
\end{array}
$$
Setting $x_n=(n-1)X_n$, we have
$$
x_n=\left\{\begin{array}{ll}x_{n-1}
 - 4(n-2) (H_{ \frac{n-2}{2}}-1)&\ \mbox{ if $n$ is even}
\\
x_{n-1}+4(n-1)H_{\lfloor \frac{n-1}{4}\rfloor}- 4(n-2) H_{ \frac{n-3}{2}} &\\
 \qquad\qquad +
8\big((\frac{n-1}{2} \mmod 2)-1\big)
&\ \mbox{ if $n$ is odd}\\
\end{array}\right.
$$
Iterating these recurrences we obtain that:
\begin{itemize}
\item If $n$ is even
$$
\begin{array}{rl} 
x_n& = x_{n-2}+4(n-2)H_{\lfloor \frac{n-2}{4}\rfloor}-4(n-3)H_{ \frac{n-4}{2}}\\
 & \qquad\qquad +8\big((\frac{n-2}{2} \mmod 2)-1\big)-4(n-2)(H_{ \frac{n-2}{2}}-1)\\
& = x_{n-2}+4(n-2)H_{\lfloor \frac{n-2}{4}\rfloor}-4(n-3)H_{ \frac{n-4}{2}}\\
 & \qquad\qquad +8\big((\frac{n-2}{2} \mmod 2)-1\big)-4(n-2)(H_{ \frac{n-4}{2}}+\frac{2}{n-2}-1)\\
& = x_{n-2}+4(n-2)H_{\lfloor \frac{n-2}{4}\rfloor}-4(2n-5)H_{ \frac{n-4}{2}}+4(n-6)\\
 & \qquad\qquad +8(\frac{n-2}{2} \mmod 2)
\end{array}
$$

\item If $n$ is odd
$$
\begin{array}{rl}
x_n & = x_{n-2}-4(n-3)(H_{ \frac{n-3}{2}}-1)+4(n-1)H_{\lfloor \frac{n-1}{4}\rfloor}\\
 & \qquad\qquad -4(n-2)H_{ \frac{n-3}{2}} +8\big((\frac{n-1}{2} \mmod 2)-1\big)\\
 & = x_{n-2}+4(n-1)H_{\lfloor \frac{n-1}{4}\rfloor}-4(2n-5)H_{ \frac{n-3}{2}}+4(n-5)\\
 & \qquad\qquad+8(\frac{n-1}{2} \mmod 2)
\end{array}
$$
\end{itemize}
From these recurrences, and noticing that $x_1=x_2=0$, we obtain that, for every $m\geq 1$
$$
\begin{array}{l}
x_{2m}  \displaystyle = \sum_{k=1}^{m-1}\Big(8kH_{\lfloor \frac{k}{2}\rfloor}-4(4k-1)H_{k-1}+8(k-2)+8(k  \mmod 2)\Big)\\
 \quad \displaystyle = 8\sum_{k=1}^{m-1} kH_{\lfloor \frac{k}{2}\rfloor} 
 -16\sum_{k=1}^{m-1} (k-1) H_{k-1}-12\sum_{k=1}^{m-1}  H_{k-1}+
 8\sum_{k=1}^{m-1}(k-2)\\
 \qquad\quad\displaystyle +8\sum_{k=1}^{m-1} (k  \mmod 2)
\end{array}
$$
$$
\begin{array}{l}
\quad \displaystyle = 8\sum_{k=1}^{m-1}kH_{\lfloor \frac{k}{2}\rfloor}
 -16\sum_{k=1}^{m-2} k  H_{k}-12\sum_{k=1}^{m-2}  H_{k}+
 4(m-1)(m-4)+8\Big\lfloor \frac{m}{2} \Big\rfloor\\
\quad \displaystyle = 
4m (m-1) H_{\lfloor \frac{m}{2} \rfloor}-8\Big\lfloor \frac{m}{2} \Big\rfloor \Big(\Big\lfloor \frac{m}{2} \Big\rfloor -1 \Big)-4(m-1)(2m-1) H_{m-1}
 \\
 \displaystyle \quad\qquad +4(m-1)(2m-3)
\\[2ex]
x_{2m+1} \displaystyle  = \sum_{k=1}^{m} \Big(8kH_{\lfloor \frac{k}{2}\rfloor}-4(4k-3)H_{k-1}+8(k-2)+8(k  \mmod 2)\Big)\\
\quad \displaystyle = 8\sum_{k=1}^{m}kH_{\lfloor \frac{k}{2}\rfloor}
-4\sum_{k=1}^{m}(4k-3)H_{k-1}+8\sum_{k=1}^{m}(k-2)+8\sum_{k=1}^{m}(k  \mmod 2)\\
\quad \displaystyle = 8\sum_{k=1}^{m}kH_{\lfloor \frac{k}{2}\rfloor}
-16\sum_{k=1}^{m-1}kH_{k}-4\sum_{k=1}^{m-1} H_{k}+4m(m-3)+8\Big\lfloor \frac{m+1}{2} \Big\rfloor\\
\quad \displaystyle = 4m (m+1) H_{\lfloor \frac{m+1}{2} \rfloor}-8\Big\lfloor \frac{m+1}{2} \Big\rfloor \Big(\Big\lfloor \frac{m+1}{2} \Big\rfloor -1 \Big)-4 m (2 m-1) H_m
\\
 \quad\qquad\displaystyle+4 m( 2 m-3)
\end{array}
$$
Both formulas correspond to:
$$
\begin{array}{rl}
x_n & \displaystyle =4\Big \lfloor \frac{n+1}{2}\Big \rfloor \Big(\Big \lfloor \frac{n+1}{2}\Big \rfloor-1)\Big) H_{\lfloor \frac{n+1}{4}\rfloor}-2(n-1)(n-2) H_{\lfloor \frac{n-1}{2}\rfloor}\\[2ex] & \qquad\qquad \displaystyle -8\Big\lfloor \frac{n+1}{4}\Big\rfloor\Big(\Big\lfloor\frac{n+1}{4}\Big\rfloor-1\Big)+
4\Big \lfloor \frac{n-1}{2}\Big \rfloor \Big(2\Big \lfloor \frac{n}{2}\Big \rfloor-3 \Big) 
\end{array}
$$
Finally,
$$
\begin{array}{rl}
X_n & =\displaystyle\frac{1}{n-1} x_n\\[2ex]
& \displaystyle = \frac{1}{n-1}\Big\{4\Big \lfloor \frac{n+1}{2}\Big \rfloor \Big(\Big \lfloor \frac{n+1}{2}\Big \rfloor-1)\Big) H_{\lfloor \frac{n+1}{4}\rfloor}-2(n-1)(n-2) H_{\lfloor \frac{n-1}{2}\rfloor}\\[1.5ex]  &\qquad\qquad \displaystyle -8\Big\lfloor \frac{n+1}{4}\Big\rfloor\Big(\Big\lfloor\frac{n+1}{4}\Big\rfloor-1\Big)+
4\Big \lfloor \frac{n-1}{2}\Big \rfloor \Big(2\Big \lfloor \frac{n}{2}\Big \rfloor-3 \Big)\Big\}
\end{array}
$$

\item Set 
$$
y_n=\sum_{k=1}^{n-2} 
E_Y(C_k)H_{\lfloor \frac{n-k-1}{2}\rfloor}=\sum_{k=1}^{n-2} (k(H_{\lfloor \frac{k}{2}\rfloor}-1)+(k\mmod 2))H_{\lfloor \frac{n-k-1}{2}\rfloor}
$$
so that $Y_n=2y_n/(n-1)$.
If $n=2m$, then
$$
\begin{array}{l}
y_{2m}   \displaystyle = \sum_{j=1}^{m-1} 2j(H_j-1)H_{m-j-1}+ \sum_{j=0}^{m-2} ((2j+1)(H_j-1)+1)H_{m-j-1}\\
\quad =\displaystyle \sum_{j=1}^{m-2} 2j(H_j-1)H_{m-j-1}+ \sum_{j=1}^{m-2} (2j(H_j-1)+H_j)H_{m-j-1}\\
\quad =\displaystyle 4\sum_{j=1}^{m-2} jH_jH_{m-j-1}+ \sum_{j=1}^{m-2} H_j H_{m-j-1}-4\sum_{j=1}^{m-2} jH_{m-j-1}\\
\quad =\displaystyle 4\sum_{j=1}^{m-2} jH_jH_{m-j-1}+ \sum_{j=1}^{m-2} H_j H_{m-j-1}-4\sum_{j=1}^{m-2} (m-1-j)H_{j}\end{array}
$$
$$
\begin{array}{l}
\quad =\displaystyle 4\sum_{j=1}^{m-2} jH_jH_{m-j-1}+ \sum_{j=1}^{m-2} H_j H_{m-j-1}-4(m-1)\sum_{j=1}^{m-2} H_{j}+4\sum_{j=1}^{m-2} jH_{j}\\
\quad \displaystyle = 4\binom{m}{2}(H_m^2-H_m^{(2)}-2H_{m}+2)+m(H_m^2-H_m^{(2)}-2H_{m}+2)\\
\quad\qquad \displaystyle -4(m-1)^2(H_{m-1}-1)+(m-1)(m-2)(2H_{m-1}-1)\\
\quad \displaystyle  =m(2m-1)(H_m^2-H_m^{(2)})-2m(3m-2)H_{m}+m(7m-5)
\end{array}
$$
If $n=2m+1$, then
$$
\begin{array}{l}
y_{2m+1}   \displaystyle = \sum_{j=1}^{m-1} 2j(H_j-1)H_{m-j}+ \sum_{j=0}^{m-1} ((2j+1)(H_j-1)+1)H_{m-j-1}\\
\quad =\displaystyle 2\sum_{j=1}^{m-1} jH_jH_{m-j}-2\sum_{j=1}^{m-1} jH_{m-j}+ 2\sum_{j=1}^{m-2} jH_j H_{m-j-1}\\
\quad \displaystyle\qquad -2 \sum_{j=1}^{m-2} jH_{m-j-1}+ \sum_{j=1}^{m-2} H_jH_{m-j-1}\\
\quad =\displaystyle 2\sum_{j=1}^{m-1} jH_jH_{m-j}+ 2\sum_{j=1}^{m-2} jH_j H_{m-j-1}+ \sum_{j=1}^{m-2} H_jH_{m-j-1}\\
\quad \displaystyle\qquad 
-2\sum_{j=1}^{m-1} (m-j)H_{j}-2 \sum_{j=1}^{m-2} (m-1-j)H_{j}\\
\quad =\displaystyle 2\sum_{j=1}^{m-1} jH_jH_{m-j}+ 2\sum_{j=1}^{m-2} jH_j H_{m-j-1}+ \sum_{j=1}^{m-2} H_jH_{m-j-1}\\
\quad \displaystyle\qquad 
-2(2m-1)\sum_{j=1}^{m-2}  H_{j}+4 \sum_{j=1}^{m-2} jH_{j}-2H_{m-1}\\
\end{array}
$$
$$
\begin{array}{l}
\quad =\displaystyle  2\binom{m+1}{2}(H_{m+1}^2-H_{m+1}^{(2)}-2H_{m+1}+2)\\
\quad \displaystyle\qquad+2\binom{m}{2}(H_m^2-H_m^{(2)}-2H_{m}+2)\\
\quad \displaystyle\qquad +m(H_m^2-H_m^{(2)}-2H_{m}+2)-2(2m-1)(m-1)(H_{m-1}-1)\\
\quad \displaystyle\qquad +(m-1)(m-2)(2H_{m-1}-1)-2H_{m-1}\\
\quad =\displaystyle  m(2m+1)(H_{m}^2-H_{m}^{(2)})-6m^2H_{m}+m(7m-1)\hspace*{2cm}
\end{array}
$$
Both formulas can be summarized into
$$
\begin{array}{l}
\displaystyle y_n  = \Big\lfloor \frac{n}{2}\Big\rfloor\Big(2\Big\lfloor \frac{n-1}{2}\Big\rfloor+1\Big)(H_{\lfloor \frac{n}{2}\rfloor}^2-H_{\lfloor \frac{n}{2}\rfloor}^{(2)})\\
\displaystyle\qquad-2\Big\lfloor \frac{n}{2}\Big\rfloor\Big(\Big\lfloor \frac{n}{2}\Big\rfloor+2\Big\lfloor \frac{n-1}{2}\Big\rfloor\Big)H_{\lfloor \frac{n}{2}\rfloor}+\Big\lfloor \frac{n}{2}\Big\rfloor\Big(3\Big\lfloor \frac{n}{2}\Big\rfloor +4\Big\lfloor \frac{n-1}{2}\Big\rfloor-1\Big)
\end{array}
$$
and therefore
$$
\begin{array}{l}
\displaystyle Y_n=\frac{2}{n-1}\Big\{\Big\lfloor \frac{n}{2}\Big\rfloor\Big(2\Big\lfloor \frac{n-1}{2}\Big\rfloor+1\Big)(H_{\lfloor \frac{n}{2}\rfloor}^2-H_{\lfloor \frac{n}{2}\rfloor}^{(2)})\\
\displaystyle\qquad-2\Big\lfloor \frac{n}{2}\Big\rfloor\Big(\Big\lfloor \frac{n}{2}\Big\rfloor+2\Big\lfloor \frac{n-1}{2}\Big\rfloor\Big)H_{\lfloor \frac{n}{2}\rfloor}+\Big\lfloor \frac{n}{2}\Big\rfloor\Big(3\Big\lfloor \frac{n}{2}\Big\rfloor +4\Big\lfloor \frac{n-1}{2}\Big\rfloor-1\Big)\Big\}
\end{array}
$$
\end{enumerate}

We can return now to our recursive formula for $E_Y(C_n^2)$, which now becomes
$$
\begin{array}{l}
\displaystyle E_Y(C_n^2) 
= \frac{n}{n-1}  E_Y(C_{n-1}^2)+
4(n-2)H_{\lfloor \frac{n-1}{2}\rfloor}-3(n-2)\\
\qquad \displaystyle+ \frac{4(n-2)}{n-1}((n-1) \mmod 2)\\
\qquad \displaystyle+ \frac{1}{n-1}\Big\{4\Big \lfloor \frac{n+1}{2}\Big \rfloor \Big(\Big \lfloor \frac{n+1}{2}\Big \rfloor-1)\Big) H_{\lfloor \frac{n+1}{4}\rfloor}-2(n-1)(n-2) H_{\lfloor \frac{n-1}{2}\rfloor}\\  \qquad\qquad \displaystyle -8\Big\lfloor \frac{n+1}{4}\Big\rfloor\Big(\Big\lfloor\frac{n+1}{4}\Big\rfloor-1\Big)+
4\Big \lfloor \frac{n-1}{2}\Big \rfloor \Big(2\Big \lfloor \frac{n}{2}\Big \rfloor-3 \Big)\Big\}\\
 \qquad \displaystyle+\frac{2}{n-1}\Big\{\Big\lfloor \frac{n}{2}\Big\rfloor\Big(2\Big\lfloor \frac{n-1}{2}\Big\rfloor+1\Big)(H_{\lfloor \frac{n}{2}\rfloor}^2-H_{\lfloor \frac{n}{2}\rfloor}^{(2)})\\
\displaystyle\qquad\qquad-2\Big\lfloor \frac{n}{2}\Big\rfloor\Big(\Big\lfloor \frac{n}{2}\Big\rfloor+2\Big\lfloor \frac{n-1}{2}\Big\rfloor\Big)H_{\lfloor \frac{n}{2}\rfloor}+\Big\lfloor \frac{n}{2}\Big\rfloor\Big(3\Big\lfloor \frac{n}{2}\Big\rfloor +4\Big\lfloor \frac{n-1}{2}\Big\rfloor-1\Big)\Big\}\\

\quad \displaystyle = \frac{n}{n-1}  E_Y(C_{n-1}^2)+
(n-2)(2H_{\lfloor \frac{n-1}{2}\rfloor}-3)\\
\qquad \displaystyle+ \frac{1}{n-1}\Big\{4\Big \lfloor \frac{n+1}{2}\Big \rfloor \Big(\Big \lfloor \frac{n+1}{2}\Big \rfloor-1)\Big) H_{\lfloor \frac{n+1}{4}\rfloor}  -8\Big\lfloor \frac{n+1}{4}\Big\rfloor\Big(\Big\lfloor\frac{n+1}{4}\Big\rfloor-1\Big) \\
 \qquad \qquad \displaystyle-12\Big \lfloor \frac{n-1}{2}\Big \rfloor+2\Big\lfloor \frac{n}{2}\Big\rfloor\Big(2\Big\lfloor \frac{n-1}{2}\Big\rfloor+1\Big)(H_{\lfloor \frac{n}{2}\rfloor}^2-H_{\lfloor \frac{n}{2}\rfloor}^{(2)}) \\
 \qquad \qquad \displaystyle
 -4\Big\lfloor \frac{n}{2}\Big\rfloor\Big(\Big\lfloor \frac{n}{2}\Big\rfloor+2\Big\lfloor \frac{n-1}{2}\Big\rfloor\Big)H_{\lfloor \frac{n}{2}\rfloor}+2\Big\lfloor \frac{n}{2}\Big\rfloor\Big(3\Big\lfloor \frac{n}{2}\Big\rfloor +8\Big\lfloor \frac{n-1}{2}\Big\rfloor-1\Big)\\
\displaystyle\qquad\qquad+ 4(n-2)\big((n-1) \mmod 2\big)\Big\}
\end{array}
$$
Setting $z_n=E_Y(C_n^2)/n$, this identity becomes
$$
\begin{array}{l}
\displaystyle z_n
= z_{n-1}+
\frac{n-2}{n}(2H_{\lfloor \frac{n-1}{2}\rfloor}-3)\\
\qquad \displaystyle+ \frac{1}{n(n-1)}\Big\{4\Big \lfloor \frac{n+1}{2}\Big \rfloor \Big(\Big \lfloor \frac{n+1}{2}\Big \rfloor-1)\Big) H_{\lfloor \frac{n+1}{4}\rfloor}  -8\Big\lfloor \frac{n+1}{4}\Big\rfloor\Big(\Big\lfloor\frac{n+1}{4}\Big\rfloor-1\Big) \\
 \qquad \qquad \displaystyle-12\Big \lfloor \frac{n-1}{2}\Big \rfloor+2\Big\lfloor \frac{n}{2}\Big\rfloor\Big(2\Big\lfloor \frac{n-1}{2}\Big\rfloor+1\Big)(H_{\lfloor \frac{n}{2}\rfloor}^2-H_{\lfloor \frac{n}{2}\rfloor}^{(2)}) \\
 \qquad \qquad \displaystyle
 -4\Big\lfloor \frac{n}{2}\Big\rfloor\Big(\Big\lfloor \frac{n}{2}\Big\rfloor+2\Big\lfloor \frac{n-1}{2}\Big\rfloor\Big)H_{\lfloor \frac{n}{2}\rfloor}+2\Big\lfloor \frac{n}{2}\Big\rfloor\Big(3\Big\lfloor \frac{n}{2}\Big\rfloor +8\Big\lfloor \frac{n-1}{2}\Big\rfloor-1\Big)\\
\displaystyle\qquad\qquad+ 4(n-2)\big((n-1) \mmod 2\big)\Big\}
\end{array}
$$
Writing this equation as $z_n=z_{n-1}+f(n)$, its solution with $z_1=0$ is
$$
z_n=\sum_{k=2}^n f(k)
$$
and it remains to compute this sum.
To do it, we split it into eight terms,
$$
\begin{array}{l}
\displaystyle S_1(n) = \sum_{k=2}^n \frac{k-2}{k}\left(2 H_{\lfloor \frac{k-1}{2}\rfloor}-3\right)
\\
\displaystyle S_2(n)= 4\sum_{k=2}^n \frac{1}{k(k-1)} \Big\lfloor\frac{k+1}{2}\Big\rfloor \left(\Big\lfloor\frac{k+1}{2}\Big\rfloor -1\right) H_{\lfloor\frac{k+1}{4}\rfloor}
\\
\displaystyle S_3(n) = 8\sum_{k=2}^n \frac{1}{k(k-1)} \Big\lfloor\frac{k+1}{4}\Big\rfloor \left(\Big\lfloor\frac{k+1}{4}\Big\rfloor -1\right)
\end{array}
$$
$$
\begin{array}{l}
\displaystyle S_4(n)= 12 \sum_{k=2}^n \frac{1}{k(k-1)} \Big\lfloor\frac{k-1}{2}\Big\rfloor 
\\
\displaystyle S_5(n)=2 \sum_{k=2}^n \frac{1}{k (k-1)} \Big\lfloor\frac{k}{2}\Big\rfloor \left( 2\Big\lfloor \frac{k-1}{2}\Big\rfloor +1\right) \left( H_{\lfloor\frac{k}{2}\rfloor}^2 -H_{\lfloor\frac{k}{2}\rfloor}^{(2)}\right)
\\
\displaystyle S_6(n) = 4 \sum_{k=2}^n \frac{1}{k (k-1)} \Big\lfloor\frac{k}{2}\Big\rfloor \left(\Big\lfloor\frac{k}{2}\Big\rfloor +2\Big \lfloor \frac{k-1}{2}\Big\rfloor\right) H_{\lfloor\frac{k}{2}\rfloor}
\\
\displaystyle S_7(n) = 2 \sum_{k=2}^n \frac{1}{k (k-1)} \Big\lfloor \frac{k}{2}\Big\rfloor \left( 3 \Big\lfloor \frac{k}{2}\Big\rfloor + 8 \Big\lfloor \frac{k-1}{2}\Big\rfloor -1 \right)
\\
\displaystyle S_8(n) = 4\sum_{k=2}^n \frac{1}{k (k-1)} (k-2) ((k-1) \mmod 2)
\end{array}
$$
in such a way that
$$
z_n=S_1(n)+S_2(n)-S_3(n)-S_4(n)+S_5(n)-S_6(n)+S_7(n)+S_8(n).
$$
Now, we compute each one of these sums. To simplify the results, set
\[
\nosum_m = \sum_{l=1}^{m-1} \frac{H_l}{2l+1}.
\]

\noindent\textbf{\emph{Sum~$S_1$}}.
We consider  two cases, depending on the parity of $n$.
 If $n$ is even, say $n=2m$, then
$$
\begin{array}{l}
\displaystyle S_1(2m) =
\sum_{l=1}^m \frac{l-1}{l}\left(2 H_{l-1}-3\right) +
\sum_{l=1}^{m-1} \frac{2l-1}{2l+1}\left(2 H_{l}-3\right) \\
\displaystyle\quad = \sum_{l=1}^{m-1} \left(2-\frac{1}{l+1}-\frac{2}{2l+1}\right)\left(2 H_{l}-3\right) \\
\displaystyle\quad = -10 m-3+4m H_m +6 H_{2m}-(H_m^2-H_m^{(2)})-4 \nosum_m
\end{array}
$$
If $n$ is odd, say $n=2m+1$, then
$$
\begin{array}{l}
S_1(2m+1) \displaystyle= S_1(2m)+\frac{2m-1}{2m+1}(2H_{m}-3)\\
\quad\displaystyle= -10 m-3+4m H_m +6 H_{2m+1}-\frac{6}{2m+1}-(H_m^2-H_m^{(2)})-4 \nosum_m\\
\qquad\qquad \displaystyle
+\frac{4m-2}{2m+1}H_{m}-\frac{6m-3}{2m+1}
 \\
 \quad\displaystyle= -10 m-6+\Big(4m+2-\frac{4}{2m+1}\Big) H_m +6 H_{2m+1}-(H_m^2-H_m^{(2)})-4 \nosum_m
 \end{array}
$$
Both formulas agree with
$$
\begin{array}{rl}
S_1(n)& \displaystyle =-3n-3-4\Big\lfloor\frac{n}{2}\Big\rfloor+\Big(2n-4+\frac{8}{n}\Big\lfloor\frac{n}{2}\Big\rfloor\Big)H_{\lfloor\frac{n}{2}\rfloor}+6 H_{n}\\
& \qquad\displaystyle -(H_{\lfloor\frac{n}{2}\rfloor}^2-H_{\lfloor\frac{n}{2}\rfloor}^{(2)})-4 \nosum_{\lfloor\frac{n}{2}\rfloor}
\end{array}
$$

\noindent\textbf{\emph{Sum $S_2$}}.
We consider  four cases, depending of the class of $n$ mod $4$. If $n=4m-2$, then
$$
\begin{array}{l}
\displaystyle S_2(4m-2) \displaystyle =4 \Big(\sum_{l=1}^m \frac{l-1}{4l-3} H_{l-1}+\sum_{l=1}^{m-1} \frac{l}{4l-1} H_{l}+\sum_{l=1}^{m-1} \frac{2l-1}{2(4l-1)} H_{l}\\
\quad\qquad\quad\qquad\qquad\displaystyle
+\sum_{l=1}^{m-1} \frac{2l+1}{2(4l+1)} H_{l}\Big) \\
\quad \displaystyle =4 \Big(\sum_{l=1}^{m-1} \frac{l}{4l+1} H_{l}+\sum_{l=1}^{m-1} \frac{l}{4l-1} H_{l}+\sum_{l=1}^{m-1} \frac{2l-1}{2(4l-1)} H_{l}+\sum_{l=1}^{m-1} \frac{2l+1}{2(4l+1)} H_{l}\Big)\\
\quad \displaystyle =4 \sum_{l=1}^{m-1} \Big(\frac{l}{4l+1} + \frac{l}{4l-1} + \frac{2l-1}{2(4l-1)} + \frac{2l+1}{2(4l+1)}\Big) H_{l}\\
\quad \displaystyle = 4 \sum_{l=1}^{m-1}H_{l}=4m(H_m-1)
\end{array}
$$
Now, if $n=4m-1$, then
$$
S_2(4m-1)=
S_2(4m-2)+4\frac{2m(2m-1)}{(4m-1)(4m-2)} H_m=\frac{16m^2}{4m-1} H_m -4m
$$
If $n=4m$, then
$$
S_2(4m)=
S_2(4m-1)+4\frac{2m(2m-1)}{(4m-1)4m} H_m=(4m+2)H_m-4m
$$
And finally, if $n=4m+1$, then
$$
S_2(4m+1) =S_2(4m)+4\frac{2m(2m+1)}{4m(4m+1)} H_m= \frac{(4m+2)^2}{4m+1} H_m -4m
$$
These four formulas agree with
$$
S_2(n)=\Big(n+3-\frac{2}{n}\Big\lfloor\frac{n}{2}\Big\rfloor\Big)H_{\lfloor \frac{n+2}{4}\rfloor}-4\Big\lfloor \frac{n+2}{4}\Big\rfloor.
$$

\noindent\textbf{\emph{Sum $S_3$}}. We consider the same four cases as in the previous sum. If $n=4m-2$, then
$$
\begin{array}{l}
\displaystyle S_3(4m-2) = 8\left(\sum_{l=1}^m \frac{(l-1)(l-2)}{(4l-2)(4l-3)} + \sum_{l=1}^{m-1} \frac{l(l-1)}{(4l-1)(4l-2)} \right. \\ 
\displaystyle \qquad \qquad \left. + \sum_{l=1}^{m-1} \frac{l(l-1)}{4l (4l-1)}+  \sum_{l=1}^{m-1} \frac{l(l-1)}{4l (4l+1)}\right) \\
\displaystyle \qquad = 8\left(\sum_{l=1}^{m-1} \frac{l(l-1)}{(4l+2)(4l+1)} + \sum_{l=1}^{m-1} \frac{l(l-1)}{(4l-1)(4l-2)} + \sum_{l=1}^{m-1} \frac{l(l-1)}{4l (4l-1)}\right. \\ \displaystyle \qquad \qquad  \left. +  \sum_{l=1}^{m-1} \frac{l(l-1)}{4l (4l+1)}\right) \end{array}
$$
$$
\begin{array}{l}\displaystyle  \qquad = 8 \sum_{l=1}^{m-1} \frac{l(l-1)}{4l^2-1}=
 \sum_{l=1}^{m-1}\Big(2-\frac{1}{2l-1}-\frac{3}{2l+1}\Big)\\
\displaystyle \qquad = 2H_{m-1} - 4H_{2m-2}+\frac{4 m^2-4}{2m-1} = 2H_{m-1} - 4H_{2m-1}+\frac{4 m^2}{2m-1}
\end{array}
$$
If $n=4m-1$, then
$$
\begin{array}{l}
\displaystyle S_3(4m-1)= S_3(4m-2)+\frac{8m(m-1)}{(4m-1)(4m-2)}\\ 
 \displaystyle \qquad =2H_{m-1} - 4H_{2m-1}+\frac{4 m^2}{2m-1}+\frac{8m(m-1)}{(4m-1)(4m-2)}\\ 
  \displaystyle \qquad = 
2 H_{m-1} -4 H_{2m-1}+\frac{4m(2m+1)}{4m-1} 
\end{array}
$$
If $n=4m$, then
$$
\begin{array}{l}
\displaystyle S_3(4m) = S_3(4m-1)+\frac{8m(m-1)}{4m(4m-1)}\\ 
 \displaystyle \qquad = 2 H_{m-1} -4 H_{2m-1}+\frac{4m(2m+1)}{4m-1}  +\frac{2(m-1)}{4m-1}\\ 
  \displaystyle \qquad =
2 H_m -4 H_{2m}+2(m+1)
\end{array}
$$
And if $n=4m+1$, then
$$
\begin{array}{l}
\displaystyle S_3(4m+1) = S_3(4m)+\frac{8m(m-1)}{4m(4m+1)}\\   \displaystyle \qquad = 2 H_m -4 H_{2m}+2(m+1)  +\frac{2(m-1)}{4m+1}\\   \displaystyle \qquad =
2 H_m -4 H_{2m}+\frac{4m(2m+3)}{4m+1}
\end{array}
$$
These four formulas correspond to
$$
S_3(n)=2H_{\lfloor\frac{n}{4}\rfloor}-4H_{\lfloor\frac{n}{2}\rfloor}
+\frac{4}{n}\Big\lfloor\frac{n+2}{4}\Big\rfloor\Big(n+2-2\Big\lfloor\frac{n+2}{4}\Big\rfloor\Big)
$$

\noindent\textbf{\emph{Sum $S_4$}}. If $n =2m$, then
$$
\begin{array}{l}
\displaystyle S_4(2m) = 6 \left( \sum_{l=1}^m \frac{l-1}{l (2l-1)} + \sum_{l=1}^{m-1} \frac{1}{2l+1} \right)
\\   \displaystyle \qquad=6 \sum_{l=1}^{m-1}\left( \frac{l}{(l+1) (2l+1)} + \frac{1}{2l+1} \right)
 = 6 \sum_{l=1}^{m-1} \frac{1}{l+1}= 6H_m -6
\end{array}
$$
  If $n =2m+1$, then
$$
\begin{array}{l}
\displaystyle S_4(2m+1) = S_4(2m)+12\cdot \frac{m}{2m(2m+1)}=6 (H_m -1)+ \frac{6}{2m+1}\\   \displaystyle \qquad =
6H_m-\frac{12m}{2m+1}
\end{array}
$$
Both formulas agree with
$$
S_4(n)=6H_{\lfloor\frac{n}{2}\rfloor}-\frac{12}{n}\cdot \Big\lfloor\frac{n}{2}\Big\rfloor
$$

\noindent\textbf{\emph{Sum $S_5$}}.  If $n =2m$, then
$$
\begin{array}{l}
\displaystyle S_5(2m) = 2\sum_{l=1}^m \frac{j(2j-1)}{2j(2j-1)}(H_l^2 - H_l^{(2)})+2\sum_{l=1}^{m-1}\frac{j(2j+1)}{2j(2j+1)}(H_l^2 - H_l^{(2)})\\ 
 \displaystyle \qquad = \sum_{l=1}^m (H_l^2 - H_l^{(2)})+\sum_{l=1}^{m-1}  (H_l^2 - H_l^{(2)}) \\ 
  \displaystyle \qquad = 2\sum_{l=1}^{m-1}  (H_l^2 - H_l^{(2)}) + H_m^2 - H_m^{(2)} \\ 
    \displaystyle \qquad = 
(2m+1)(H_m^2 - H_m^{(2)}) -4m H_m +4m
\end{array}
$$
 If $n =2m+1$, then
$$
\begin{array}{l}
\displaystyle S_5(2m+1) = S_5(2m)+2\cdot \frac{m(2m+1)}{2m(2m+1)}(H_m^2-H_m^{(2)})\\  \displaystyle \qquad =  
(2m+2)(H_m^2 - H_m^{(2)}) -4m H_m +4m
\end{array}
$$
This shows that
$$
S_5(n)=(n+1)(H_{\lfloor\frac{n}{2}\rfloor}^2-H_{\lfloor\frac{n}{2}\rfloor}^{(2)})-4 \Big\lfloor\frac{n}{2}\Big\rfloor(H_{\lfloor\frac{n}{2}\rfloor}-1)
$$

\noindent\textbf{\emph{Sum $S_6$}}.  If $n =2m$, then
$$
\begin{array}{l}
\displaystyle S_6(2m) = 2 \left(\sum_{l=1}^m \frac{3l-2}{2l-1} H_l + \sum_{l=1}^{m-1} \frac{3l}{2l+1} H_l\right) \\ \displaystyle \qquad = 
\sum_{l=1}^{m-1} \left(6- \frac{1}{2l-1}-\frac{3}{2l+1} \right) H_l + \frac{2(3m-2)}{2m-1} H_m \\ \displaystyle \qquad = 
6m(H_m-1)-\Big(\nosum_m-\frac{4m-1}{2m-1}H_m+2H_{2m}\Big)-3\nosum_m+\frac{6m-4}{2m-1}H_m\\
\displaystyle \qquad = (6m+5) H_m -4 \nosum_m -2 H_{2m}-6m
\end{array}
$$
 If $n =2m+1$, then
$$
\begin{array}{l}
\displaystyle S_6(2m+1) = S_6(2m)+ \frac{6m}{2m+1}H_m\\
\displaystyle \qquad = (6m+5) H_m -4 \nosum_m -2 H_{2m}-6m+ \frac{6m}{2m+1}H_m\\
\displaystyle \qquad = 
\left(6m+5+\frac{6m}{2m+1}\right) H_m -4 \nosum_m -2 H_{2m+1}+\frac{2}{2m+1}-6m
\end{array}
$$
Both formulas correspond to
$$
S_6(n)=\Big(3n+2+\frac{6}{n}\Big\lfloor\frac{n}{2}\Big\rfloor\Big)H_{\lfloor\frac{n}{2}\rfloor}-4\nosum_{\lfloor\frac{n}{2}\rfloor}-2H_n-6\Big\lfloor\frac{n}{2}\Big\rfloor+ \frac{2}{n}\Big(n-2\Big\lfloor\frac{n}{2}\Big\rfloor\Big)
$$

\noindent\textbf{\emph{Sum $S_7$}}.
If $n=2m$,
$$
\begin{array}{l}
\displaystyle S_7(2m) = \sum_{l=1}^m \frac{11l-9}{2l-1}+\sum_{l=1}^{m-1} \frac{11l-1}{2l+1} \\ \displaystyle \qquad =\sum_{l=1}^{m-1} \left(11-\frac{1}{2}\left(\frac{7}{2l-1}+\frac{13}{2l+1} \right) \right)  
 + \frac{11m-9}{2m-1}  \\ \displaystyle \qquad  = 11 (m-1)-10 \sum_{l=1}^m \frac{1}{2l-1}+\frac{7}{2(2m-1)}+\frac{13}{2}+\frac{11m-9}{2m-1} \\ \displaystyle \qquad  = 11 m+1+ 5 H_m-10 H_{2m}  
\end{array}
$$
If $n=2m+1$,
$$
\begin{array}{l}
\displaystyle S_7(2m+1) = S_7(2m)+2\cdot \frac{m(11m-1)}{2m(2m+1)}\\
\displaystyle \qquad = 11 m+1+ 5 H_m-10 H_{2m} +\frac{11m-1}{2m+1}\\
\displaystyle \qquad =
5 H_m-10 H_{2m+1}+ 11m+1+\frac{11m+9}{2m+1}
\end{array}
$$
Both formulas agree with
$$
S_7(n)=5H_{\lfloor\frac{n}{2}\rfloor}-10H_n+5n+\Big\lfloor\frac{n}{2}\Big\rfloor+1+\frac{1}{2}\Big(1+\frac{7}{n}\Big)\Big(n-2\Big\lfloor\frac{n}{2}\Big\rfloor\Big)
$$

\noindent\textbf{\emph{Sum $S_8$}}. If $n=2m$, 
$$
S_8(2m) = 4 \sum_{l=1}^m \frac{l-1}{l(2l-1)} = 4  \sum_{l=1}^m \left( \frac{1}{l}-\frac{1}{2l-1}\right)  = 6 H_m -4  H_{2m},
$$
and if $n=2m+1$,
$$
S_8(2m+1) = S(2m)+\frac{(2m-1)\cdot 0}{2m(2m+1)} = 6 H_m -4  H_{2m} =6H_m-4H_{2m+1}+\frac{4}{2m+1}
$$
Therefore,
$$
S_8(n)=6H_{\lfloor\frac{n}{2}\rfloor}-4H_n+\frac{4}{n}\Big(n-2\Big\lfloor\frac{n}{2}\Big\rfloor\Big)
$$

Now, once we have computed $S_1,\ldots, S_8$, we can compute $z_n$:
$$
\begin{array}{rl}
z_n & = S_1(n)+S_2(n)-S_3(n)-S_4(n)+S_5(n)-S_6(n)+S_7(n)+S_8(n)\\
& \displaystyle =\frac{5n+7}{2}+\Big(6+\frac{1}{n}\Big)\Big\lfloor\frac{n}{2}\Big\rfloor
+\frac{8}{n}\Big\lfloor\frac{n+2}{4}\Big\rfloor^2-\frac{8(n+1)}{n}\Big\lfloor\frac{n+2}{4}\Big\rfloor
\\
&\quad\displaystyle +n(H_{\lfloor\frac{n}{2}\rfloor}^2-H_{\lfloor\frac{n}{2}\rfloor}^{(2)})-6H_n+\Big(3-n-\Big(4-\frac{2}{n}\Big)\Big\lfloor\frac{n}{2}\Big\rfloor\Big)H_{\lfloor\frac{n}{2}\rfloor}\\
&\quad\displaystyle +\Big(n+3-\frac{2}{n}\Big\lfloor\frac{n}{2}\Big\rfloor\Big)H_{\lfloor\frac{n+2}{4}\rfloor}-2H_{\lfloor\frac{n}{4}\rfloor}
\end{array}
$$

And finally
$$
 \displaystyle \begin{array}{l}
E_Y(C_n^2)=nz_n\\
\quad\displaystyle  =\frac{5n^2+7n}{2}+(6n+1)\Big\lfloor\frac{n}{2}\Big\rfloor
+8\Big\lfloor\frac{n+2}{4}\Big\rfloor^2-8(n+1)\Big\lfloor\frac{n+2}{4}\Big\rfloor
\\
\qquad\displaystyle +n^2(H_{\lfloor\frac{n}{2}\rfloor}^2-H_{\lfloor\frac{n}{2}\rfloor}^{(2)})-6nH_n+\Big(3n-n^2-(4n-2)\Big\lfloor\frac{n}{2}\Big\rfloor\Big)H_{\lfloor\frac{n}{2}\rfloor}\\
\qquad\displaystyle +\Big(n^2+3n-2\Big\lfloor\frac{n}{2}\Big\rfloor\Big)H_{\lfloor\frac{n+2}{4}\rfloor}-2nH_{\lfloor\frac{n}{4}\rfloor}
\end{array}
$$
as we claimed. 
\hspace*{\fill}\qed
\end{proof}

\begin{corollary}\label{th:varC}
The variance of $C_n$ under the Yule model is
$$
  \begin{array}{l}
\sigma_Y^2(C_n^2)  \displaystyle  =\frac{5n^2+7n}{2}+(6n+1)\Big\lfloor\frac{n}{2}\Big\rfloor
-4\Big\lfloor\frac{n}{2}\Big\rfloor^2+8\Big\lfloor\frac{n+2}{4}\Big\rfloor^2\\
\qquad\displaystyle -8(n+1)\Big\lfloor\frac{n+2}{4}\Big\rfloor
-6nH_n+\Big(2\Big \lfloor\frac{n}{2}\Big\rfloor - n(n-3)\Big)H_{\lfloor\frac{n}{2}\rfloor}\\
\qquad\displaystyle -n^2H_{\lfloor\frac{n}{2}\rfloor}^{(2)}+\Big(n^2+3n-2\Big\lfloor\frac{n}{2}\Big\rfloor\Big)H_{\lfloor\frac{n+2}{4}\rfloor}-2nH_{\lfloor\frac{n}{4}\rfloor}
\end{array}
$$
\end{corollary}

\begin{proof}
Simply replace in the formula $\sigma^2_Y(C_n)=E_Y(C_n^2)-E_Y(C_n)^2$ the value of $E_Y(C_n^2)$ obtained in the last theorem and the value of 
$$
E_Y(C_n)=n(H_{\lfloor \frac{n}{2}\rfloor}-1)+(n \mmod 2)=nH_{\lfloor \frac{n}{2}\rfloor}-2\Big\lfloor \frac{n}{2}\Big\rfloor
$$ 
recalled above. \hspace*{\fill}\qed
\end{proof}

\begin{corollary}
$$
\begin{array}{l}
\displaystyle \sigma^2_Y(C_n)=
-\frac{8}{3} \Big(-18+\pi ^2+\log (64)\Big)\Big \lfloor \frac{n}{4}\Big\rfloor^2 -8 \Big\lfloor\frac{n}{4}\Big\rfloor\log\Big(\Big\lfloor\frac{n}{4}\Big\rfloor\Big) \\ \qquad\displaystyle+
\Big(20-8\gamma-32\log(2) +\Big(24-\frac{4}{3}\pi^2-8\log(2)\Big) (n \mmod 4)\Big)\Big \lfloor\frac{n}{4}\Big\rfloor+ O(1).
\end{array}
   $$
   \end{corollary}

\begin{proof}
We expand the expression for $\sigma^2_Y(C_n)$ given in the previous corollary, 
taking into account the value of~$n$ module~$4$.

If there exists some $m$ such that $n=4m$, then 
$$
\begin{array}{l}
\sigma^2_Y(C_n)=8 m \left(2 m H_m-2 (m-1) H_{2 m}-3 H_{4 m}-2 m H_{2 m}^{(2)}+6 m+1\right)\\
\quad\displaystyle =-\frac{8}{3} m^2 \left(-18+\pi ^2+\log (64)\right)+m \left(-8 \log (m)-8 \gamma +20-32 \log (2)\right) \\ \qquad\displaystyle -2+O\left( \frac{1}{m}\right)
\end{array}
$$
  
If there exists some $m$ such that $n=4m+1$, then 
$$
\begin{array}{l}
\sigma^2_Y(C_n)=  2 \left(8 m^2+4 m+1\right) H_m+\left(-16 m^2+8 m+2\right) H_{2 m}-24 m H_{4 m+1}\\ \qquad -6 H_{4 m+1}-16 m^2 H_{2 m}^{(2)}-8 m H_{2 m}^{(2)}-H_{2 m}^{(2)}+48 m^2+32 m+6\\
\quad\displaystyle=   -\frac{8}{3} m^2 \left(-18+\pi ^2+\log (64)\right)-\frac{4}{3} m (6  \log (m)+\pi ^2+6 \gamma -33+30 \log (2)) \\ \qquad\displaystyle  +\left(-2  \log (m)-\frac{\pi ^2}{6}-2 \gamma +4-10 \log (2)\right)+O\left(\frac{1}{m}\right)
\end{array}
$$
  
If there exists some $m$ such that $n=4m+2$, then 
$$
\begin{array}{l}
\sigma^2_Y(C_n)= 16 m H_{2 m+2}-24 m H_{4 m+2}+4 (2 m+1)^2 H_{m+1}-4 (2 m+1)^2 H_{2 m+1} \\ \qquad\displaystyle +8 H_{2 m+2} -12 H_{4 m+2}-16 m^2 H_{2 m+1}^{(2)}-16 m H_{2 m+1}^{(2)}-4 H_{2
   m+1}^{(2)}\\ \qquad+48 m^2+48 m+10\\
   \quad\displaystyle=   -\frac{8}{3} m^2 \left(-18+\pi ^2+\log (64)\right)+m \left(-8 \log (m)-\frac{8 \pi ^2}{3}-8 \gamma +68-48 \log (2)\right)\\ \qquad\displaystyle -\frac{2}{3} (6
   \log (m)+\pi ^2+6 \gamma -24+30 \log (2)) +O\left(\frac{1}{m}\right)
  \end{array}
$$
  
Finally, if there exists some $m$ such that $n=4m+3$, then 
$$
\begin{array}{l}
\displaystyle\sigma^2_Y(C_n)= (4 m+3) \big[4 (m+1) H_{m+1}-4 (m+1) H_{2 m+1}+4 H_{2 m+2}-6 H_{4 m+3} \\  \qquad\displaystyle - (4 m+3) H_{2 m+1}^{(2)}  +10 m+11\big] +(2 m+1) \left(-2 H_{m+1}+2 H_{2
   m+1}+24 m+19\right) \\  \qquad\displaystyle-24 (m+1)^2-4 (2 m+1)^2\\
    \quad\displaystyle=    -\frac{8}{3} m^2 \left(-18+\pi ^2+\log (64)\right)-4 m \left(2  \log (m)+\pi ^2+2 \gamma -23+7 \log (4)\right) \\  \qquad\displaystyle +\left(-6 \log (m)-\frac{3 \pi ^2}{2}-6 \gamma +34-17 \log (4)\right)+O\left(\frac{1}{m}\right)
  \end{array}
$$

Now, using that $m=\lfloor n/4\rfloor$ and 
$n\mod 4=n-4\lfloor n/4\rfloor$, it can be easily seen that
 these formulas are consistent with the 
development of $\sigma^2_Y(C_n)$ until $O(1)$ given in the statement. \hspace*{\fill}\qed
\end{proof}

\section{Some covariances}

In this section we obtain the covariance under the Yule model of $S_n$ and $\Phi_n$ from the formulas obtained in the previous sections for  $E_Y(\oF^2_n)$, $E_Y(S_n^2)$ and $E_Y(\Phi_n^2)$.

\begin{corollary}
$ \Cov_Y(S_n, \Phi_n)= 4n(nH_n^{(2)}+H_n)+\frac{1}{6} n (n^2-51 n+2)$.
\end{corollary}

\begin{proof}
Notice that 
$$
\begin{array}{l}
\Cov_Y(S_n, \Phi_n)  =E_Y(S_n\cdot \Phi_n)-E_Y(S_n)\cdot E_Y(\Phi_n)\\
\qquad \displaystyle =\frac{1}{2}\big(E_Y((\Phi_n+S_n)^2)-E_Y(S_n^2)-E_Y(\Phi_n^2)\big)-E_Y(S_n)\cdot E_Y(\Phi_n)\\[2ex]
\qquad \displaystyle =\frac{1}{2}\big(E_Y(\oF^2_n)-E_Y(S_n^2)-E_Y(\Phi_n^2)\big)-E_Y(S_n)\cdot E_Y(\Phi_n)
\end{array}
$$
The formula in the statement is obtained by replacing in this identity 
$E_Y(\oF^2_n)$, $E_Y(S_n^2)$, $E_Y(\Phi_n^2)$, $E_Y(S_n)$, and $E_Y(\Phi_n)$ by their values. \hspace*{\fill}\qed
\end{proof}

\begin{corollary}
$\displaystyle \Cov_Y(S_n, \oF_n)= 2nH_n+\frac{1}{6} n (n^2-9 n-4)$
\end{corollary}

\begin{proof}
By the bilinearity of covariances,
$\Cov_Y(S_n, \oF_n)=\Cov_Y(S_n, S_n+\Phi_n)=\sigma_Y^2(S_n)+\Cov_Y(S_n, \Phi_n)$. \hspace*{\fill}\qed
\end{proof}

\begin{corollary}
$$
\begin{array}{rl}
\displaystyle \Cov_Y(S_n, \Phi_n) & = \displaystyle \frac{1}{6}n^3+\Big(\frac{2\pi^2}{3}-\frac{17}{2}\Big)n^2+4n\ln(n)+\frac{1}{3}(12\gamma-11)n +4+O\Big(\frac{1}{n}\Big)\\[2ex]
\Cov_Y(S_n, \oF_n) & =
 \displaystyle \frac{1}{6}n^3-\frac{3}{2}n^2+2n\ln(n)+\frac{1}{3}(6\gamma-2)n +1+O\Big(\frac{1}{n}\Big)
 \end{array}
 $$
\end{corollary}

From the formulas for $\sigma_Y^2(\Phi_n)$, $\sigma_Y^2(\oF_n)$, and $\Cov_Y(S_n, \Phi_n)$, we can compute Pearson's correlation coefficient between $S_n$ and $\Phi_n$,
$$
cor_Y(S_n,\Phi_n)=\frac{\Cov_Y(S_n, \Phi_n)}{\sqrt{\sigma_Y^2(\Phi_n)\cdot \sigma_Y^2(\oF_n)}}.
$$
The exact formula for this coefficient is
$$
cor_Y(S_n,\Phi_n)\textstyle =\frac{4n(nH_n^{(2)}+H_n)+\frac{1}{6} n (n^2-51 n+2)}{\sqrt{\big(7n^2-4n^2H_n^{(2)}-2nH_n-n\big)\big(\frac{1}{12}(n^4-10n^3+131 n^2-2n)-4n^2 H_n^{(2)}-6nH_n\big)}}
$$
and in the limit it is equal to
$$
cor_Y(S_n,\Phi_n)\sim \frac{1}{6\sqrt{\big((7-\frac{2\pi^2}{3})\cdot \frac{1}{12} \big)}}=0.89059
$$

\section{Conclusions}

In this paper we have obtained exact formulas for the variance under the Yule model of the Colless index $C$, the Sackin index $S$, the total cophenetic index $\Phi$, and the sum $\oF=S+\Phi$, as well as for the covariances of $S$ and $\Phi,\oF$. Our formulas are explicit  and hold on spaces $\TT_n$ of binary phylogenetic trees with any number $n$ of leaves, unlike other expressions  published so far in the literature, which were either recursive or  asymptotic.

The proofs consist of elementary, although long and involved,  algebraic computations. Since it is not difficult to slip some mistake in such long algebraic computations, to double-check the results we have directly computed these variances and covariances on $\TT_n$ for $n=3,\ldots,9$ and confirmed that our formulas give the right results. The values obtained are given in the next table. The Python scripts used to compute them are available at the Supplementary Material web page  \url{http:/bioinfo.uib.es/~recerca/phylotrees/Yulevariances/}.

\begin{table}[htb]
\begin{tabular}{r|ccccc}
& 3& 4 & 5 & 6 & 7 \\
\hline
$\sigma_Y^2(C_n)$ & 0 & 2 & 3.5 & 6.8 & 10.072222
\\
$\sigma_Y^2(S_n)$ & 0 & 0.222222 & 0.805556 & 1.84 & 3.877778 
\\
$\sigma_Y^2(\Phi_n)$ & 0 & 0.888889 & 5.138889 & 17.04 & 42.787778  
\\
$\sigma_Y^2(\oF_n)$ & 0 & 2 & 10 & 30 & 70 
\\
$\Cov_Y(S_n, \Phi_n)$ & 0 &  0.444444 &  2.0277778  & 5.56 & 11.912222 
\\
$cor_Y(S_n,\Phi_n)$ & - & 1 & 0.996639 & 0.992958 & 0.989408 \\
\hline\hline
& 8 & 9 \\
\hline
$\sigma_Y^2(C_n)$ & 15.765079 & 21.089881
\\
$\sigma_Y^2(S_n)$ & 5.49424 & 8.193827
\\
$\sigma_Y^2(\Phi_n)$ & 90.522812 &  170.350969 
\\
$\sigma_Y^2(\oF_n)$ & 140 & 252
\\
$\Cov_Y(S_n, \Phi_n)$ & 21.991474 & 36.727602
\\
$cor_Y(S_n,\Phi_n)$ & 0.986101 & 0.983053\\

\end{tabular}
\caption{Values of $\sigma_Y^2(C_n)$, $\sigma_Y^2(S_n)$, $\sigma_Y^2(\Phi_n)$, $\sigma_Y^2(\oF_n)$, $\Cov_Y(S_n, \Phi_n)$, and $cor_Y(S_n,\Phi_n)$ for $n=3,\ldots,9$. They agree with those given by our formulas.}
\end{table}

It can be seen in this table that the values of the variances of $S_n$ are smaller than those of the variance of $\Phi$ or $\oF$. Actually, as we have recalled in the Introduction, $\sigma_Y^2(S_n)$ has order $O(n^2)$, while  
$\sigma_Y^2(\Phi_n)$ and $\sigma_Y^2(\oF_n)$ are $O(n^4)$. This is consistent with the fact that 
$\Phi$ and $\oF$ have larger spans of values than $S$, $O(n^3)$ instead of $O(n^2)$, and much less ties.
It is also deduced from the formulas obtained in this paper, and from this table for small values of $n$, that  there is a strong direct linear correlation between $S_n$ and $\Phi_n$, although in the limit Pearson's coefficient between them decreases to 0.89. 

It remains to compute  exact formulas for covariances of $C$ with $S$ and $\Phi$. These formulas
 can surely be  obtained using a recurrence for the expected value of the product of two recursive shape indices similar in spirit to  Corollary \ref{cor:YI}, but  the computations  seem to be even longer than those leading to  the computation of $\sigma_Y^2(C_n)$. 

\section*{Acknowledgements}  This  research has been partially supported by the Spanish government and the UE FEDER program, through projects MTM2009-07165 and TIN2008-04487-E/TIN.  We thank J. Mir\'o and M. Lewis for several comments on a previous version of this work. Most computations in this paper have been carried out or checked with the aid of \textsl{Mathematica}.

\end{document}